\documentclass[twoside,11pt]{article}
\usepackage{jmlr2e}
\usepackage[utf8]{inputenc}
\usepackage{amsmath}
\usepackage{amsfonts}
\usepackage{mathtools}
\usepackage{url}
\usepackage{graphicx}
\usepackage{tikz}
\usepackage[vlined]{algorithm2e}
\def\checkmark{\tikz\fill[scale=0.4](0,.35) -- (.25,0) -- (1,.7) -- (.25,.15) -- cycle;} 
\usepackage{longtable}

\ShortHeadings{Prior Editing for Graph Filter Posterior Fairness}{Krasanakis, Papadopoulos, Kompatsiaris and Symeonidis}
\firstpageno{1}

\begin{document}

\title{Prior Signal Editing \\for Graph Filter Posterior Fairness Constraints}

\author{\name Emmanouil Krasanakis \email maniospas@iti.gr\\
       \name Symeon Papadopoulos \email papadop@iti.gr\\
       \name Ioannis Kompatsiaris \email ikom@iti.gr\\
       \addr 
       Information Technologies Institute\\
       Centre for Research and Technology---Hellas\\
       57001 Thermi, Thessaloniki, Greece
       \AND
       \name Andreas Symeonidis \email asymeon@eng.auth.gr\\
       \addr 
       Department of Electrical and Computer Engineering\\
       Aristotle University of Thessaloniki\\
       54124, Thessaloniki, Greece
       }

\editor{This is an author preprint}

\maketitle

\begin{abstract}%
Graph filters are an emerging paradigm that systematizes information propagation in graphs as transformation of prior node values, called graph signals, to posterior scores. In this work, we study the problem of mitigating disparate impact, i.e. posterior score differences between a protected set of sensitive nodes and the rest, while minimally editing scores to preserve recommendation quality. To this end, we develop a scheme that respects propagation mechanisms by editing graph signal priors according to their posteriors and node sensitivity, where a small number of editing parameters can be tuned to constrain or eliminate disparate impact. We also theoretically explain that coarse prior editing can locally optimize posteriors objectives thanks to graph filter robustness. We experiment on a diverse collection of 12 graphs with varying number of nodes, where our approach performs equally well or better than previous ones in minimizing disparate impact and preserving posterior AUC under fairness constraints.
\end{abstract}


\section{Introduction}
Relational data can be organized into graphs, where entities are represented as nodes and linked through edges corresponding to real-world relations between them. 
Due to the pervasiveness of complex graphs across disciplines such as social networking, epidemiology, genomic analysis and software engineering, various schemes have been proposed to mine relational information. These range from the unsupervised paradigms of node clustering \citep{schaeffer2007graph,kulis2009semi} and exposing underlying structures through edge sparsification \citep{spielman2004nearly} to semi-supervised inference of posterior attribute scores based on known node prior values \citep{kipf2016semi,chen2020simple}. A mechanism favored by many approaches is propagating latent or predictive attribute information through graphs via recursive aggregation among node neighbors. For example, information propagation has been used in unsupervised extraction of tightly knit node clusters with few outgoing edges \citep{andersen2006local,wu2012learning}, node ranking algorithms that recommend nodes based on their structurally proximity to a set of query ones \citep{tong2006fast,kloster2014heat} and recent graph neural network advances that decouple latent attribute extraction with their propagation to neighbors \citep{klicpera2018predict,dong2020equivalence,huang2020combining}.
\par
A systematic way to study information propagation is through graph signal processing \citep{gavili2017shift,ortega2018graph,sandryhaila2013discrete}. This domain extends discrete signal processing to higher dimensions, where signals comprise prior values spread not across points in time but across graph nodes.\footnote{\footnotesize Discrete signal processing can be modeled by graph signal processing if time is considered a line graph where points in time are nodes and consecutive ones are linked.} In analogy to time filters, graph filters produce posterior node scores through a weighted aggregation of propagating priors different hops away, where propagation follows either spectral or stochastic rules.
\par
Fairness concerns arise when the outputs of data mining systems, such as graph filters, are correlated to sensitive attributes, such as gender or ethnicity \citep{chouldechova2017fair,kleinberg2018algorithmic}. Previous research has studied bias mitigation in the sense that sensitive and non-sensitive groups of data samples behave similarly under evaluation measures of choice \citep{chouldechova2017fair,krasanakis2018adaptive,zafar2019fairness,ntoutsi2020bias}. In this work, we tackle the fairness objective of achieving (approximate) statistical parity between sensitive and non-sensitive positive predictions---a concept known as disparate impact elimination \citep{biddle2006adverse,calders2010three,kamiran2012data,feldman2015certifying} and often assessed through a measure called \emph{pRule}. In particular, we explore the problem of imposing fairness objectives on graph filter posteriors, such as maximimizing pRule or making it reach a predetermined level, while maintaining the ability to form accurate predictions over different thresholding criteria, as measured by recommender system measures \citep{shani2011evaluating,wang2013theoretical,isinkaye2015recommendation}, such as AUC. For instance, this can help ensure that a protected set of nodes (the ones considered sensitive) are also frequently but not erroneously recommended.
\par
This work extends our previous paper on making graph filter posteriors fair \citep{krasanakis2020applying}. There, we proposed that satisfying fairness-aware objectives with minimal impact to posterior score quality can be achieved by appropriately editing graph signal priors. This way, new posteriors can be guided to be fairer while respecting information propagation through the graph's structure, which is responsible for predictive quality. In our previous paper, we first tackled this objective by editing priors with schemes of few parameters, which depend on whether nodes are sensitive and the differences between priors and posteriors. Our base assumption was that parameters can be tuned to yield priors proportionate to their estimated contribution to fair yet similar to original posteriors.
\par
In this work we improve various aspects of our previous research. First, we provide a novel mathematical framework to express how tightly prior editing mechanisms should approximate the gradients of posterior objectives and use it to explain why coarse prior editing models can locally optimize fairness-aware objectives. Second, we propose an alternative to our approach that uses an error-based instead of perturbation-based surrogate model, a more robust training objective that is not impacted by high posterior score outliers and an explicit prior retention term. These changes induce higher average AUC when mitigating disparate impact, though the two alternatives outperform each other on different experiments. Finally, we assess the efficacy of our work by experimenting on a significantly larger corpus of 12 instead of 4 multidisciplinary real-world graphs combined with 8 instead of 2 base graph filters. To facilitate an informed discussion of our results, we also perform a rigorous instead of empirical post-hoc analysis of summary statistics.

\section{Background}\label{background}
In this section we provide the theoretical background necessary to understand our work. We start with a common community-based interpretation of node scores and their practical usefulness in Subsection~\ref{community}. Then, in Subsection~\ref{gsp}, we present graph signal processing concepts used to study a wide range of methods for obtaining node scores given prior information of node attributes. We finally discuss algorithmic fairness under the prism of graph mining and overview the limited research done to merge these disciplines in Subsection~\ref{fairness background}. The operations and symbols used in this work are summarized in Table~\ref{notation}.

\begin{table}[hbt]
\centering
\footnotesize
\label{notation}
\begin{tabular}{l|l}
    \textbf{Notation} & \textbf{Interpretation}\\
    \hline 
        $\mathcal{I}$ & Identity matrix with appropriate dimensions\\
        $\textbf{0}$ & Column vector of appropriate rows and zero elements\\
        $\textbf{1}$ & Column vector of appropriate rows and one elements\\
        $r[v]$ & Element corresponding to node $v$ of graph signal $r$\\
        $\mathcal{L}(r)$ & Loss function for graph filter posteriors $r$\\
        $\nabla \mathcal{L}(r)$ & Gradient vector of loss $\mathcal{L}(r)$ with elements $\nabla \mathcal{L}(r)[v]=\frac{\partial \mathcal{L}(r)}{\partial r[v]}$\\
        $|x|$ & Absolute value for numbers, number of elements for sets\\
        $\|x\|$ & L2 norm of vector $x$ computed as $\sqrt{\sum_v x[v]^2}$\\
        $\|x\|_1$ & L1 norm of vector $x$ computed as $\sum_v |x[v]|$\\
        $\|x\|_\infty$ & Maximum value of $x$ computed as $\max_v x[v]$\\
        $\lambda_1$ & Smallest eigenvalue of a positive definite matrix\\
        $\lambda_{\max}$ & Largest eigenvalue of a positive definite matrix\\
        $H(W)$ & Graph filter on normalization $W$ of the adjacency matrix\\
        $A\setminus B$ & Set difference, that is the elements of $A$ not found in $B$\\
        $a^T b$ & Dot product of column vectors $a,b$ as matrix multiplication\\
        $\mathbb{R}^{|\mathcal{V}|}$ & Space of column vectors comprising all graph nodes\\
        $diag([\lambda_i]_i)$ & A diagonal matrix $diag([\lambda_i]_i)[i,j]=\{\lambda_i\text{ if }i=j,0\text{ otherwise}\}$\\
        $A[u,v]$ & Element of matrix $A$ at row $u$ and column $v$\\
        $A^T$ & Transposition of matrix $A$ for which $A^T[u.v]=A[v,u]$\\
        $A^{-1}$ & Inverse of invertible matrix $A$\\
        $\{x\,|\,cond(x)\}$ & Elements $x$ satisfying a condition $cond$\\
        $P(e)$ & Probability of event $e$ occurring\\
        $P(e|cond)$ & Probability of event $e$ given that condition $cond$ is satisfied\\
        $\mathcal{S}$ & Set of sensitive nodes\\
        $\mathcal{S}'$ & Complement of $S$, that is the set of non-sensitive nodes\\
\end{tabular}
\caption{Mathematical notation. Graph-related quantities refer to a common studied graph.}
\end{table}

\subsection{Node Scores and Community Structure}\label{community}
Nodes of real-world graphs can often be organized into communities of either ground truth structural characteristics \citep{fortunato2016community,leskovec2010empirical,xie2013overlapping,papadopoulos2012community} or shared node attributes \citep{hric2014community,hric2016network,peel2017ground}. A common task in graph analysis is to score or rank all nodes based on their relevance to such communities. This is particularly important for large graphs, where community boundaries can be vague \citep{leskovec2009community,lancichinetti2009detecting}. Furthermore, node scores can be combined with other characteristics, such as their past values when discovering nodes of emerging importance in time-evolving graphs, in which case they should be of high quality across the whole graph. Many algorithms that discover communities with only a few known members also rely on transforming and thresholding node scores \citep{andersen2006local,whang2016overlapping}.
\par
A measure frequently used to quantify the quality of attribute recommendations, such as node scores, is AUC  \citep{hanley1982meaning}, which compares operating characteristic trade-offs at different decision thresholds. If we consider ground truth node scores $q_{test}[v]=\{1\text{ if node }v\text{ is a community member},0\text{ otherwise}\}$ (these form a type of graph signal defined in the next subsection), and posterior scores $r[v]$, the True Positive Rate (TPR) and False Positive Rate (FPR) operating characteristics for decision thresholds $\theta$ can be respectively defined as:
\begin{align*}
   &TPR(\theta)=P(q_{test}[v]=1\,|\,r[v]\geq\theta)\\
   &FPR(\theta)=P(q_{test}[v]=0\,|\,r[v]\geq\theta)
\end{align*}
where $P(a|b)$ denotes the probability of $a$ conditioned on $b$. Then, AUC is defined as the cumulative effect induced to the TPR when new decision thresholds are used to change the FPR and quantified per the following formula:
\begin{equation}\label{AUC}
    AUC=\int_{-\infty}^\infty TPR(\theta) FPR'(\theta)\,d\theta
\end{equation}
AUC values closer to $100\%$ indicate that community members achieve higher scores compared to non-community members, whereas $50\%$ AUC corresponds to randomly assigned node scores.

\subsection{Graph Signal Processing}\label{gsp}
\subsubsection{Graph signals} 
Graph signal processing \citep{ortega2018graph} is a domain that extends traditional signal processing to graph-structured data. To do this, it starts by defining graph signals $q:\mathcal{V}\to\mathbb{R}$ as maps that assign real values $q[v]$ to graph nodes $v\in\mathcal{V}$.\footnote{\footnotesize Signals with multidimensional node values can be expressed as ordered collections of real-valued signals.} Graph signals can be represented as column vectors $q'\in\mathbb{R}^{|\mathcal{V}|}$ with elements $q'[i]=q[\mathcal{V}[i]]$, where $|\cdot|$ is the number of set elements and $\mathcal{V}[i]$ is the $i$-th node of the graph after assuming an arbitrary fixed order. For ease of notation, in this work we use graph signals and their vector representations interchangeably by replacing nodes with their ordinality index---in other words, we assume the isomorphism $\mathcal{V}[i]= i$.
\par
Graph signals are often constructed from sets of query nodes that share an attribute of interest, in which case their elements are assigned binary values depending on whether respective nodes are \emph{queries} $q[v]=\{1\text{ if }v\text{ has the attribute}, 0\text{ otherwise}\}$. For example, the attribute of interest could capture whether nodes belong to the same structural community as query ones. In general, we consider graph signals that are constrained to non-negative values $q\in[0,\infty)^{|\mathcal{V}|}$, where $q[v]$ are proportional to the probabilities that nodes $v$ exhibit the attribute of interest. In this case, graph signal elements can be understood as node rank scores.
\subsubsection{Graph signal propagation}
A pivotal operation in graph signal processing is the one-hop propagation of node values to their graph neighbors, where incoming values are aggregated on each node. Expressing this operation for unweighted graphs with edges $\mathcal{E}\subseteq\mathcal{V}\times\mathcal{V}$ requires the definition of adjacency matrices $A$, whose elements correspond to the binary existence of respective edges, i.e. $A[i,j]=\{1\text{ if }(\mathcal{V}[i], \mathcal{V}[j])\in\mathcal{E}, 0\text{ otherwise}\}$. A normalization operation is typically employed to transform adjacency matrices into new ones $W$ with the same dimensions, but which model some additional assumptions about the propagation mechanism (see below for details). Then, single-hop propagation of node values stored in graph signals $q$ to neighbors yields new graph signals $q_{next}$ with elements $q_{next}[u]=\sum_{v\in\mathcal{V}}W[u,v]q[v]$. For the sake of brevity, this operation is usually expressed using linear algebra as $q_{next}=Wq$.
\par
Two popular types of adjacency matrix normalization are column-wise and symmetric. The first sets up one-hop propagation as a stochastic process \citep{tong2006fast} that is equivalent to randomly walking the graph and selecting the next node to move to from a uniformly random selection between neighbors. Formally, this is expressed as $W_{col}=AD^{-1}$, where $D=diag\big(\big[\sum_{v}A[u,v]\big]_{u}\big)$ is the diagonal matrix of node degrees. Columns of the column-normalized adjacency matrix sum to $1$. This way, if graph signals priors model probability distributions over nodes, i.e. their values sum to $1$ and are non-negative, posteriors also model probability distributions.
\par
On the other hand, symmetric normalization arises from a signal processing perspective, where the eigenvalues of the normalized adjacency matrix are treated as the graph's spectrum \citep{chung1997spectral,spielman2012spectral}. In this case---and if the graph is undirected in that the existence of edges $(u,v)$ also implies the existence of edges $(v,u)$---a symmetric normalization is needed to guarantee that eigenvalues are real numbers. To achieves this, the normalization $W_{symm}=D^{-1/2}AD^{-1/2}$ is predominantly selected, on merit that it has bounded eigenvalues and $D^{1/2}(I-W_{symm})D^{1/2}$ is a Laplacian operator that implements the equivalent of discrete derivation over graph edges.
\par
Graph signal processing research often targets undirected graphs and symmetric normalization, since these enable computational tractability and closed form convergence bounds of resulting tools.\footnote{\footnotesize Theoretical groundwork has been established to consider spectral equivalent for undirected graphs and, ultimately, asymmetric adjacency matrix normalizations \citep{chung2005laplacians,yoshida2019cheeger}. Unfortunately, these involve the extraction of Peron vectors in polynomial yet non-linear times \citep{bjorner1992chip} that do not scale well to graphs with hundreds of thousands or millions of nodes and edges.} In this work, we also favor this practice, because it also allows the graph equivalent of signal filtering we later adopt to maintain spectral characteristics needed by our analysis. The key property we take advantage of is that, as long as the graph is connected, the normalized adjacency matrix $W$ is invertible and has the same number of real-valued eigenvalues as the number of graph nodes $|\mathcal{V}|$ residing in the range $[-1,1]$. If we annotate these eigenvalues as  $\{\lambda_i\in[-1,1]\,|\,i=1,\dots,|\mathcal{V}|\}$, the symmetric normalized adjacency matrix's Jordan decomposition takes the form: 
\begin{equation*}
    W=U^{-1} diag([\lambda_i]_i)U
\end{equation*} 
where $U$ is an orthonormal matrix with columns the corresponding eigenvectors. Therefore, scalar multiplication, power and addition operations on the adjacency matrices also transform eigenvalues the same way. For example, it holds that $W^n=U^{-1}diag([\lambda_i^n]_i)U$.
\subsubsection{Graph filters} The one-hop propagation of graph signals is a type of shift operator in the multidimentional space modeled by the graph, in that it propagates values based on a notion of relational proximity. Based on this observation, graph signal processing uses it analogously to the time shift $z^{-1}$ operator and defines the notion of graph filters as weighted aggregations of multi-hop propagations. In particular, since $W^nq$ expresses the propagation $n=0,1,2,\dots$ hops away of graph signals $q$, a weighted aggregation of these hops means that the outcome of graph filters can be expressed as:
\begin{equation}\label{filter}
\begin{split}
&r=H(W)q\\
&H(W)=\sum_{n=0}^\infty h_n W^n
\end{split}
\end{equation}
where $H(W)$ is graph filter characterized by real-valued weights $\{h_n\in\mathbb{R}|n=0,1,2,\dots\}$ indicating the importance placed on propagation of graph signals $n$ hops away. In this work, we understand the resulting graph signal $r$ to capture \textit{posterior node scores} that are the result of passing \textit{prior node values} of the original signal $q$ through the filter.
\par
In this work, we consider graph filters that are positive definite matrices, that is whose eigenvalues are all positive. For symmetric normalized graph adjacency matrices with eigenvalues $\{\lambda_i\in[-1,1]|i=1,\dots,|\mathcal{V}|\}$, it is easy to check whether graph filters defined per Equation~\ref{filter} are positive definite, as they assume eigenvalues $\big\{H(\lambda_i)=\sum_{n=0}^\infty h_n\lambda_i^n\,|\,i=1,\dots,|\mathcal{V}|\big\}$
and hence we can check whether the graph filter's corresponding polynomial assumes only positive values:
\begin{equation}\label{condition}
    H(\lambda)>0\,\forall \lambda\in[-1,1]
\end{equation}
For example, filters arising from decreasing importance of propagating more hops away $h_n>h_{n+1}\,\forall n=0,1,\dots$ are positive definite.  
Two well-known graph filters that are positive definite for symmetric adjacency matrix nomralizations are personalized pagerank \citep{andersen2007local,bahmani2010fast} and heat kernels \citep{kloster2014heat}. These respectively arise from power degradation of hop weights $h_n=(1-a)a^n$ and the exponential kernel $h_n=e^{-t}{t^n}/{n!}$ for empirically chosen parameters $a\in[0,1]$ and $t\in\{1,2,3,\dots\}$.

\subsubsection{Sweep ratio} The sweep procedure \citep{andersen2006local,andersen2007using} is a well-known mechanism that takes advantage of graph filters to identify tightly-knit congregations of nodes that are also well-separated from the rest of the graph---a concept known as low subgraph conductance \citep{chalupa2017memetic}. When attributes modeled by graph signal priors are closely correlated to the formation of structural communities, it enhances the recommendation quality of posteriors.
\par
In detail, the sweep procedure assumes that a base personalized graph node scoring algorithm $R$ with strong locality \citep{wu2012learning}, such as personalized pagerank and heat kernels, outputs graph signal posteriors $R(q)$ for graph signal priors $q$ that comprise structurally close query nodes. The posteriors are said to be personalized on the query nodes and are compared to their non-personalized counterparts $R(\textbf{1})$, where $\textbf{1}$ is a vector of ones, through the following division:
\begin{equation}\label{sweep}
    r_{sweep}[v]=\frac{R(q)[v]}{R(\textbf{1})[v]}
\end{equation}
In this work, we follow the terminology of our previous research \citep{krasanakis2020applying} and refer to the post-processing of Equation~\ref{sweep} as the \textit{sweep ratio}. The sweep procedure orders all nodes based on this ratio and splits the order in two partitions so that conductance is minimized for the respective graph cut. This practice statistically yields well-separated partitions for a variety of node ranking algorithms \citep{andersen2006local,andersen2007using,chalupa2017memetic}. Hence, when nodes are scored based on their relevance to structural communities, the sweep ratio can be deployed to improve their quality. 

\subsection{Fairness of Graph Filter Posteriors}\label{fairness background}
In this subsection we introduce the well-known concept of disparate impact in the domain of algorithmic fairness, which we aim to either fully or partially mitigate. We also overview previous works that can be used to bring fairness to graph mining and graph filter posteriors.

\subsubsection{Disparate impact elimination} Algorithmic fairness is broadly understood as parity between sensitive and non-sensitive samples over a chosen statistical property. Three popular fairness-aware objectives \citep{chouldechova2017fair,krasanakis2018adaptive,zafar2019fairness,ntoutsi2020bias} are disparate treatment elimination, disparate impact elimination and disparate mistreatment elimination. These correspond to not using the sensitive attribute in predictions, preserving statistical parity between the fraction of sensitive and non-sensitive positive labels and achieving identical predictive performance on the two groups under a measure of choice.
\par
In this work, we focus on mitigating disparate impact unfairness \citep{chouldechova2017fair,biddle2006adverse,calders2010three,kamiran2012data,feldman2015certifying}. An established measure that quantifies this objective is the \textit{pRule} \citep{biddle2006adverse}; denoting as $R[v]\in\{0,1\}$ the binary outputs of a system $R$ for samples $v$, $S$ the set of sensitive samples and $S'$ the set of non-sensitive ones, that is the complement of S:
\begin{equation}
\begin{split}
    &pRule = \frac{\min\{q_{S},q_{S'}\}}{\max\{q_{S},q_{S'}\}}\in[0,1]\\
    &q_{S} = P(R[v]=1|p\in S)\\
    &q_{S'} = P(R[v]=1|p\not\in S)
\end{split}
\end{equation}
The higher the pRule, the fairer a system is. There is precedence \citep{biddle2006adverse} for considering $80\%$ pRule or higher fair, and we also adopt this constraint in our experiments later.
\par
Calders-Verwer disparity $|p_S-q_{S'}|$ \citep{calders2010three} is also a well-known disparate impact assessment measure. However, although it is optimized at the same point as the pRule, it biases fairness assessment against high fractions of positive predictions. For example, it considers the fractions of positive labels $(p_S,q_{S'})=(0.8, 0.6)$ less fair than $(p_S,q_{S'})=(0.4, 0.3)$. We shy away from this understanding because a stochastic interpretation of posterior nodes scores could be scaled by a constant yet unknown factor. On the other hand, the pRule would quantify both numerical examples as $\tfrac{0.6}{0.8}=\tfrac{0.3}{0.4}=75\%$ fair. 
\subsubsection{Posterior score fairness}
In domains related to the outcome of graph mining algorithms, fairness has been defined for the order of recommended items \citep{beutel2019fairness,biega2018equity,yang2017measuring,zehlike2017fa} as equity in the ranking positions between sensitive and non-sensitive items. However, these notions of fairness are not applicable to the more granular understanding provided by node scores.
\par
Another definition of graph mining fairness has been introduced for node embeddings \citep{bose2019compositional,rahman2019fairwalk} under the guise of fair random walks---the stochastic process modeled by personalized pagerank when the adjacency matrix is normalized by columns. Yet the fairness of these walks is only implicitly asserted through embedding fairness. Furthermore, they require a certain number of sensitive nodes to make sure that at least one is available to walk to at every step.
\par
Fairness has also been recently explored for the outcome of graph neural networks \citep{dai2020fairgnn}, which can be trained to produce fair recommendations, even under partial knowledge of sensitive attributes. Still, advances on graph neural network theory \citep{dong2020equivalence} suggest that they can be perceived as decoupled multiplayer perceptron and graph filter components, in which case the question persists on how to make the outcome of graph filters fair, given potentially unfair priors. For graph neural networks that follow a predict-then-propagate paradigm \citep{klicpera2018predict}, this could be achieved by the aforementioned practice of adding fairness objectives to loss functions responsible for the construction of graph signal priors from data. However, the model deriving graph signal priors could be too costly to retrain or not available. In such cases, the problem of inducing fairness reverts to the graph signal processing viewpoint of this work.
\par
A recent work by \cite{tsioutsiouliklis2020fairness} has initiated a discourse on posterior score fairness. Although focused on algorithms scoring the structural importance of nodes without any kind of personalization pertaining to specific graph signal priors, it first recognized the need for optimizing a trade-off between fairness and preserving posterior score quality. Furthermore, it provided a first definition of node score fairness, called $\phi$-fairness. Under a stochastic interpretation of node scores, where they are proportional to the probability of nodes assuming positive labels, $\phi$-fairness is equivalent to full disparate impact elimination when $\phi=\frac{|S|}{|S|+|S'|}$.
\par
In our previous work \citep{krasanakis2020applying}, we introduced a generalization of pRule to posterior scores, which we also adopt in this work. In particular, we start from the above-mentioned stochastic interpretation of posteriors and calculate the expected positive number of sensitive and non-sensitive node labels obtained by sampling mechanism that uniformly assigns positive labels with probabilities proportional to posterior scores. This defines the quantities:
\begin{equation*}
    \begin{split}
        &p_S=P(R[v]=1|p\in S)=\tfrac{1}{|S|} {\sum_{v\in S} \tfrac{r[v]}{\|r\|_\infty}}\\
        &q_{S'}=P(R[v]=1|p\not\in S)=\tfrac{1}{|S'|}{\sum_{v\not\in S} \tfrac{r[v]}{\|r\|_\infty}}
    \end{split}
\end{equation*}
where the maximum value of posteriors $\|r\|_\infty=\max_u r[u]$ is used for normalization and $R$ is a stochastic process with probability $P(R[v]=1)=\frac{r[v]}{\|r[u]\|_\infty}$. Plugging these in the pRule cancels out the normalization of dividing both the nominator and denominator with the same value and yields the following formula for calculating a stochastic interpretation of the pRule given posterior node scores $r$:
\begin{equation}\label{prule}
\begin{split}
    &pRule(r) = \frac{\min\big\{|S'|\sum_{v\in S}r[v],|S|\sum_{v\not\in S}r[v]\big\}}{\max\big\{|S'|\sum_{v\in S}r[v],|S|\sum_{v\not\in S}r[v]\big\}}
\end{split}
\end{equation}
By convention, we consider $pRule=0$ when all node scores are zero.

\section{Optimizing Posterior Scores}
Optimizing graph filter posteriors to better satisfy objectives by adjusting their values deteriorates the quality gained by passing priors through graph filters. Ultimately, this leads to losing the robustness of propagating information through node relations by introducing one additional degree of freedom for each node. 
To prevent loss of quality, we propose that, instead of posteriors, graph signal priors can be edited based on their initial posteriors so as to find new posteriors better satisfying set objectives. This scheme is demonstrated in Figure~\ref{fig:overview}, where initial priors $q$ and their respective posteriors $r$ are used to construct new edited priors $q_{est}$. The editing process is ideally controlled by few parameters, which can be tuned to let posteriors $r_{est}$ of edited priors optimize the objective.
\begin{figure}[tbp]
\centering
    \includegraphics[width=0.8\textwidth,clip]{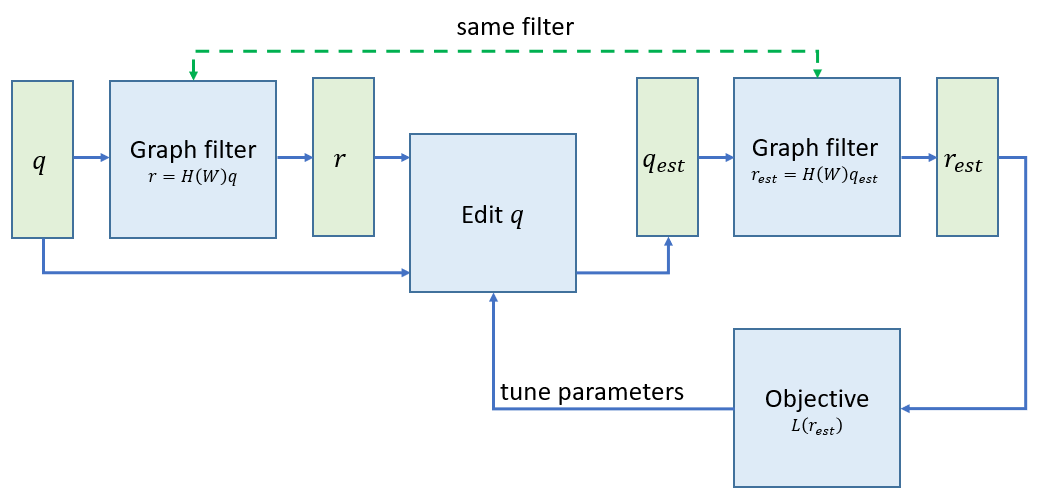}
    \caption{Starting from graph priors $q$ and estimating edited priors $q_{est}$ that help find graph filter posteriors $r_{est}$ locally optimizing an objective $\mathcal{L}(r_{est})$.}
    \label{fig:overview}
\end{figure}

\par
In Subsection~\ref{objectives} we demonstrate fairness-aware objectives that encourage disparate impact mitigation and discourage large posterior changes. Then, in Subsection~\ref{fairness aware personalization} we explore mechanisms of editing priors to be proportional to the probability of respective posterior scores approaching ideal ones optimizing the objective. This probability is estimated with a surrogate model of differences between original priors and posteriors. Given that biases against sensitive nodes could also affect the ability to estimate ideal posteriors, different model parameters are reserved for nodes with sensitive attributes. We also propose an alternative to our previous approach, that uses absolute errors instead of differences between priors and posteriors and introduces an additional parameter to explicitly control original prior retention.
The involvement of sensitive attributes in both prior editing and objective calculation is demonstrated in the updated scheme of Figure~\ref{fig:sensitive_overview}. 

\begin{figure}[tbp]
\centering
    \includegraphics[width=0.8\textwidth,clip]{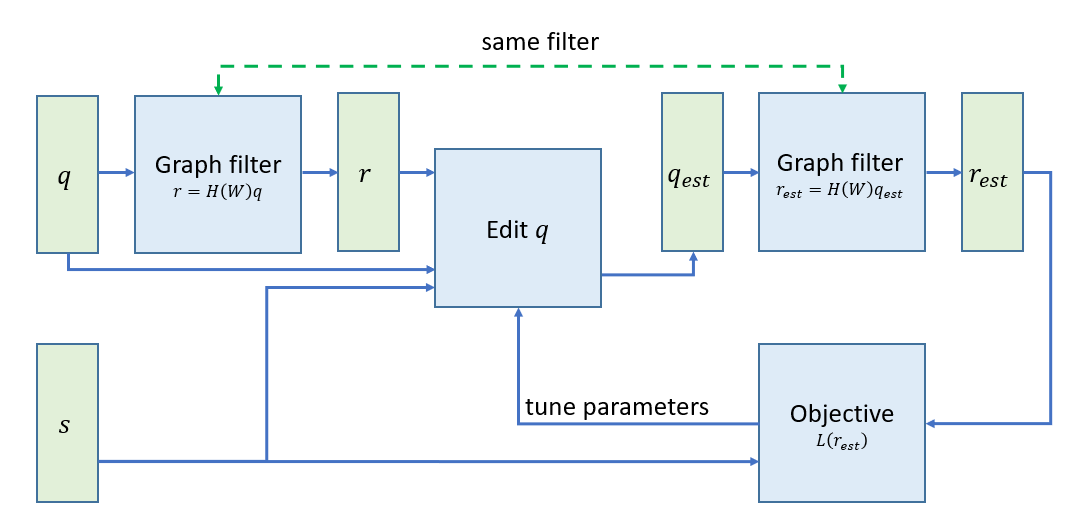}
    \caption{Starting from graph priors $q$ and estimating edited priors $q_{est}$ that help find graph filter posteriors $r_{est}$ locally optimizing an fairness-aware objective $\mathcal{L}(r_{est})$. A graph signal $s$ holding sensitive attribute values is used for prior editing and calculating the objective.}
    \label{fig:sensitive_overview}
\end{figure}
\par
Finally, in Subsection~\ref{achieving optimality} we present a new theoretical justification of why prior editing approaches are able to (locally) optimize posterior objectives using few parameters. To do this, we take advantage of graph filter robustness against prior perturbations to show in Theorem~\ref{robustness} that non-tight (that is, bound by the eigenvalue ratio of minimum over maximum eigenvalues multiplied by scalar depending on the post-processing) approximations of prior optimization slopes suffice to lead posteriors to local optimality. Under the assumption that prior editing gradients exhibits enough degrees of freedom to express objective gradients with the same loose tightness as before, we then show in Theorem~\ref{anyedit} that there exist (not necessarily observed) parameter optimization trajectories that arive at edited priors with locally optimal posteriors.

\subsection{Fairness-Aware Objectives}\label{objectives}
Tuning prior editing mechanisms requires fairness-aware objectives that trade-off disparate impact mitigation and original posterior score preservation. The first of these two components is quantified by the pRule, which can also be constrained so that values over a threshold $sup_{prule}$ do not induce further gains but lower ones are penalized with large weights $w_{prule}$. On the other hand, retaining original posterior scores can be assessed through a distance measure between the original posteriors $r$ and the estimated ones $r_{est}$. In our previous work, we penalized the mean absolute difference between the max-normalized version of original and new posteriors:

\begin{equation}\label{objective}
\begin{split}
    &\text{minimize } \mathcal{L}(r_{est})\\
    &\mathcal{L}(r_{est}) = \frac1{|V|} {\bigg\|\frac{r_{est}}{\|r_{est}\|_\infty}-\frac{r}{\|r\|_\infty}\bigg\|_1}-w_{prule}\cdot \min\{pRule(r_{est}),sup_{prule}\}
\end{split}
\end{equation}

A shortcoming of this objective is that it is heavily influenced by the maximum posterior scores used for normalization. Furthermore, if a few high posteriors were disproportionately large (for example, orders of magnitude larger) compared to the rest, their comparisons would dominate the assessment. As an extreme example, if a clique of nodes with prior values $1$ were disconnected from the rest of the graph, personalized pagerank would yield posteriors $1$ for these and hence would not normalize other node scores at all. Similar arguments hold if priors favor the discovery of well-separated communities from the rest of the graph, as graph filter algorithms often do. For instance, the sweep ratio explicitly aims to amplify this phenomenon.
\par
Dividing posteriors with summary statistics, such as their L1 norm, can not prevent high posterior scores from assuming disproportionately larger values and dominate comparisons. To address this persisting issue, in this work we move away from the mean absolute difference comparison and instead use the KL-divergence of estimated posteriors distributions given the distribution of original high quality posteriors. A clear advantage of this measure is that it compares distributions as a whole (it has thus been used to estimate differences between different Gaussian distributions for metric learning \citep{wang2017multi}) and is hence less influenced by outliers or few disproportionately large posteriors. We convert posteriors to distributions by normalizing them by division with their L1 norm, in which case the KL divergence measures the entropy lost by moving from the original posterior distribution to the new one:
\begin{equation}
KL(r_{est}\,|\,r) = -\tfrac{r_{est}^T}{\|r_{est}\|_1} \ln \tfrac{r/\|r\|_1}{r_{est}/\|r_{est}\|_1}
\end{equation}
where the vector division is performed element-by-element and by convention $0\ln x =0\,\forall x\in[0,\infty)\cup\{\infty\}$. Replacing the mean absolute error in the previous objective with KL-divergence, we end up with the following more robust variation of the objective:
\begin{equation}\label{objectiveKL}
\begin{split}
    &\text{minimize } \mathcal{L}(r_{est})\\
    &\mathcal{L}(r_{est}) = KL(r_{est}\,|\,r)-w_{prule}\cdot \min\{pRule(r_{est}),sup_{prule}\}
\end{split}
\end{equation}

\subsection{Fairness-Aware Prior Editing}\label{fairness aware personalization}
We now explore in more detail the advances of our previous work \citep{krasanakis2020applying}, where we ported to graphs a pre-processing scheme that weights training samples to make well-calibrated black box classifiers fair \citep{krasanakis2018adaptive}. This scheme proposes that unfairness is correlated to misclassification error (the difference between binary classification labels and calibration probabilities) and is influenced by whether samples are members of a sensitive group. Since strongly misclassified samples could exhibit different degrees of bias from correctly classified ones, it skews calibrated probabilities to make them fair by a transformation of misclassification error. Different skewing parameters are determined for sensitive and non-sensitive samples.
\par
These considerations introduce a type of balancing between original predictive abilities and sources of unfairness, depending on the parameters of misclassification error transformation and the ones controlling whether calibration probabilities should increase and decrease. Unfortunately, the same principles can not be ported as-are to graph signal processing, because weighting zero node priors through multiplication does not affect posteriors and there is no posterior validation set on which to tune permutation parameters. We address these issues by respectively performing non-linear edits of graph signal priors instead of scaling them and using priors as a rough one-class validation set of known positive examples.

\subsubsection{Conditional error estimation} Under the above assumptions, we refer to a stochastic interpretation of graph signal posteriors similar to the one used to set up the stochastic generalization of pRule; posteriors are snapped to $1$ with probability proportional to their value and to $0$ otherwise. We also consider edited graph signal priors $q_{est}$ that help estimate posteriors $r_{est}=H(W)q_{est}$ of similar fairness to some unobserved ideal ones $r_{fair}$. Then, given that graph signal priors suggest a ground truth categorization of nodes based on predictive attributes, we condition the probability of fairness-inducing node posteriors achieving their ideal values on whether they match their priors, that is on whether high posterior values correspond to high prior values. This is formally expressed for graph nodes $v$ as:
\begin{align*}
P(r_{fair}[v]=r_{est}[v])
&= P(r_{fair}[v]=r_{est}[v]\,|\,q[v]=r_{est}[v])P(q[v]=r_{est}[v])
\\&+P(r_{fair}[v]=r_{est}[v]\,|\,q[v]\neq r_{est}[v])P(q[v]\neq r_{est}[v])
\end{align*}
\par
Since graph filters tend to be strongly local in the sense that their posteriors preserve at least some portion of their priors (a fraction of their value equal to or greater than $\frac{h_0}{\sum_{n=0}^\infty h_n}$ comes from their priors if non-negative parameters are used to define filters per Equation~\ref{filter}), we further argue that fairness of posteriors pertains to fairness of respective priors. Hence, probabilities of estimated node posteriors approaching fair ones given tight approximations original priors are correlated with the probabilities of priors being fair $P(q_{fair}[v]=q_{est}[v])$. The same reasoning applies for estimated posteriors \textit{not} approximating well the original priors.
\par
An additional observation is that, when one of the conditional probabilities of correctly discovering fairness-aware posteriors increases (compared to the rest of nodes), the others should decrease and conversely, as they are conditioned on complementary events. Therefore, they should not simultaneously overestimate or underestimate original ones. Formally, this property can be expressed as:
\begin{align*}
    &\big(\tfrac{1}{K_{p}}P(r_{fair}[v]=r_{est}[v]\,|\,q[v]=r_{est}[v])-P(r_{fair}[v]=r_{est}[v])\big)\\
    &\quad \cdot \big(\tfrac{1}{K_{n}}P(r_{fair}[v]=r_{est}[v]\,|\,q[v]\neq r_{est}[v])- P(r_{fair}[v]=r_{est}[v])\big) < 0
\end{align*}
for some scaling parameters $K_{p},K_{n}>0$ that let probabilities sum to $1$. We now develop a surrogate model 
satisfying this property. This model depends on differences between posteriors and priors and whether nodes are sensitive:
\begin{align*}
&P(r_{fair}[v]= r_{est}|q[v]=r_{est}[v])\approx K_{p}\,P(q_{fair}[v]=q_{est}[v])\, e^{-b[v](r[v]/\|r\|_\infty-q[v])}
\\&P(r_{fair}[v]= r_{est}|q[v]\neq r_{est}[v])\approx K_{n}\,P(q_{fair}[v]=q_{est}[v])\,e^{b[v](r[v]/\|r\|_\infty-q[v])}
\end{align*}
where $b$ is a vector of real values such that $b[v]=\{b_S\text{ if }v\in S, b_{S'}\text{ otherwise}\}$. 

\subsubsection{Prior editing mechanism} The selection of sensitive or non-sensitive nodes to contribute to original priors can be viewed as Bernoulli trials with probabilities $\alpha_S$ and $\alpha_{S'}$ \citep{krasanakis2018adaptive}, which may differ depending on bias involved in prior selection. For example, there is often a lower probability of reporting sensitive nodes to construct priors \citep{cassel2019risk}. We organize prior selection probabilities into a graph signal $\alpha$ with values
$\alpha[v]=P(q[v]=r_{est}[v])=\{\alpha_S\text{ if }v\in S, \alpha_{S'}\text{ otherwise}\}\in[0,1]$.
\par
Therefore, given the surrogate model of conditional errors, we make a fair prior estimation $q_{est}$ of ideal priors based on the self-consistency criterion that, when it approaches fairness-inducing personalization, estimated fair ranks should also approach the ideal fair ones:
\begin{equation}\label{FP}
\begin{split}
    q_{est}[v]&= P(r_{fair}[v]=r_{est}[v]\,|\,q_{fair}[v]=q_{est}[v])\approx\frac{P(r_{fair}[v]=r_{est}[v])}{P(q_{fair}[v]=q_{est}[v])}
    \\&= \alpha[v] K_{p} e^{-b[v](r[v]/\|r\|_\infty-p[v])}+(1-\alpha[v])K_{n}e^{b[v](r[v]/\|r\|_\infty-p[v])}
    \\&\propto a[v] e^{-b[v](r/\|r\|_\infty-p[v])}+(1-a[v])e^{b[v](r[v]/\|r\|_\infty-p[v])}
\end{split}
\end{equation}
for graph signal $a[v]=\{a_S\text{ if }v\in S, a_{S'}\text{ otherwise}\}\in[0,1]$ that relates to prior selection probabilities and permutation scaling factors per $a[v]\propto \tfrac{\alpha[v]K_{p}}{\alpha[v]K_{p}+(1-\alpha[v])K_{n}}$. Determining $a$ requires only two independent parameters $a_S,a_{S'}$.

\subsubsection{Error-based prior editing} In the above process we started with a surrogate model of conditional fair posterior probability estimation that uses the differences between original priors and their normalized posteriors as inputs. Although this approach was met with success in our previous work, we also point out that it implicitly involves information on whether nodes assume positive values in the original binary priors instead of only the deviation between posteriors and priors. In detail, the differences $r[v]/\|r\|_\infty-p[v]$ are positive for zero priors $p[v]=0$ but negative for positive priors $p[v]=1$. For example, when $b[v]>0$ high posteriors are penalized for the conditional probability of priors approaching estimated posteriors when in reality high posteriors need to be encouraged for positive priors. Similarly, when $b[v]>0$ it encourages small posteriors for the conditional probability of priors \textit{not} approaching estimated posteriors when in reality small posteriors need to be encouraged for zero priors.
\par
To introduce an explainable interpretation of parameters $b[v]$, we thus propose the alternative of adopting the absolute value of differences by using errors $|r[v]/\|r\|_\infty-p[v]|$ as inputs to the surrogate model. Following the same analytical process as above, this leads us to the following error-based prior editing mechanism as a potential competitor:
\begin{equation*}
    q_{est}[v]= a[v] e^{-b[v]\big|r/\|r\|_\infty-p[v]\big|}+(1-a[v])e^{b[v]\big|r[v]/\|r\|_\infty-p[v]\big|}
\end{equation*}
The ability of error-based perturbations to mitigate disparate impact while enjoying high predictive quality has also been justified and experimentally corroborated in previous research \citep{krasanakis2018adaptive}. As a final step, and to provide more granular control of the retention of original posteriors required by objective functions, we introduce an explicit trade-off between original and edited fairness-aware priors by adding a term of the latter weighted by a new parameter $a_0\in[0,1]$:
\begin{equation}\label{FairEdit}
    q_{est}[v]= a_0q[v]+a[v] e^{-b[v]\big|r/\|r\|_\infty-p[v]\big|}+(1-a[v])e^{b[v]\big|r[v]/\|r\|_\infty-p[v]\big|}
\end{equation}

\subsection{Local Optimality of Prior Editing}\label{achieving optimality}
In this subsection we investigate theoretical properties of positive definite graph filters that allow coarse approximations of optimal prior editing schemes to also reach local optimality with regards to posterior objectives.

\subsubsection{Graph filters with post-processing} Real-world graph filters often end up with a post-processed version of posteriors. For example, when personalized pagerank is implemented as an iterative application of the power method scheme $r= aWr+(1-a)q$, potential feedback loops that continuously increase posteriors for some asymmetric adjacency matrix normalizations are often avoided by performing L1 normalization, which divides node posteriors with their sum. This ends up not affecting the ratio of importance scores placed on propagating prior graph signals different number of hops away, but produces a scaled version of the filter for which node score posteriors sum to $1$.
\par
More complex post-processing mechanisms, such as the sweep ratio, may induce different transformations per node that can not be modeled by graph filters. Taking these concerns into account, we hereby introduce a notation with which to formalize post-processing; we consider a \textit{post-processing vector} with which posteriors are multiplied element-by-element. Formally, this lets us write post-processed posteriors $r$ of passing graph signal priors $q$ through a graph filter $H(W)$ as: 
\begin{equation}
    r = diag(p) H(W) q
\end{equation}
We stress that the exact post-processing transformation could change depending on both the graph filter and graph signal priors. However, we can decouple this dependency by thinking of the finally selected post-processing as one particular selection out of many possible ones. In this work we consider two types of postprocessing: a) L1 output normalization and b) the sweep ratio. These are respectively  modeled as multiplication with the inverse of the original posterior's L1 norm $q_{L1}[v]=\tfrac{1}{\|H(W)q\|_1}$ and the inverse of non-personalized posteriors
$q_{sweep}[v]=\tfrac{1}{(H(W)\textbf{1})[v]}$.

\subsubsection{Coarse approximation of gradients} To formalize the concept of approximately tracking the optimization slope of posterior objectives with small enough error, in Definition~\ref{optimizer} we introduce $\lambda$-optimizers of objectives as multivariate multivalue functions that approximate the (negative) gradients of objectives with relative error bound $\lambda$. Smaller values of this strictness parameter indicate tighter approximation of optimization slopes, whereas smaller values indicate looser tracking. To disambiguate the possible directions of slopes, our analysis considers loss functions of non-negative values to be minimized. 

\begin{definition}\label{optimizer}
A continuous function $f:\mathcal{R}\to \mathbb{R}^{|\mathcal{V}|}$ will be called a $\lambda$-optimizer of a loss function $\mathcal{L}(r)$ over graph signal domain $\mathcal{R}\subseteq \mathbb{R}^{|\mathcal{V}|}$ only if: 
\begin{equation*}
    \|f\big(r\big)+\nabla \mathcal{L}(r)\|< \lambda \|\nabla \mathcal{L}(r)\|\quad \text{ for all }\nabla \mathcal{L}(r)\neq \textbf{0}
\end{equation*}.
\end{definition}
We now analyse the robustness of positive definite graph filters with post-processing in terms of how tight optimizers of posterior losses should be for the graph filter's propagation mechanism to ``absorb'' the error. To this end, in Theorem~\ref{robustness} we find a maximum tightness parameter sufficient to lead to local optimality of posteriors with respect to the loss. The required tightness depends on the graph filter's maximum and minimum eigenvalues and the post-processing vector's maximum and minimum values and requires the graph filter to be symmetric positive definite. 

\begin{theorem}\label{robustness}
Let $H(W)$ be a positive definite graph filter and $p$ a postprocessing vector. If $f(r)$ is a $\tfrac{\lambda_1 \min_v p[v]}{\lambda_{\max} \max_v p[v]}$-optimizer of a differentiable loss $\mathcal{L}(r)$ over graph signal domain $\mathcal{R}\subseteq\mathbb{R}^{|\mathcal{V}|}$, where $\lambda_1,\lambda_{\max}>0$ are the smallest positive and largest eigenvalues of $H(W)$ respectively, updating graph signals per the rule:
\begin{equation}\label{update}
\begin{split}
    &\frac{\partial q}{\partial t}=f(r)
    \\&r=diag(p)H(W)q
\end{split}
\end{equation}
asymptotically leads to the loss to local optimality if posterior updates are closed in the domain, i.e. $r\in \mathcal{R}\Rightarrow diag(p)H(W)\int_0^{t}f(r)dt\in\mathcal{R}$.
\end{theorem}
\begin{proof}
For non-negative posteriors $r=diag(p)H(W)q$, and non-zero loss gradients, the Cauchy-Shwartz inequality in the bilinear space $\langle x,y\rangle=x^TH(W)y$ determined by the positive definite graph filter $H(W)$ yields:
\begin{align*}
    &  \frac{d \mathcal{L}(r)}{dt}
    \\&\quad=(\nabla \mathcal{L}(r))^T diag(p) H(W) \frac{d q}{dt}
    \\&\quad=(\nabla \mathcal{L}(r))^T diag(p) H(W)f(r)
    \\&\quad=(\nabla \mathcal{L}(r))^T diag(p) H(W)\big( -\nabla \mathcal{L}(r)+(f(r)+\nabla\mathcal{L}(r))\big)
    \\&\quad\leq -\min_v p[v](\nabla \mathcal{L}(r))^T H(W)\nabla \mathcal{L}(r)
    +(\nabla \mathcal{L}(r))^Tdiag(p)H(W)(f(r)+\nabla\mathcal{L}(r)
    \\&\quad\leq -\lambda_1\min_v p[v]\|\nabla \mathcal{L}(r)\|^2
    +\lambda_{\max}\|f(r)+\nabla\mathcal{L}(r)\|\|diag(p)\nabla \mathcal{L}(r)\|
    \\&\quad< -\min_v p[v]\lambda_1\|\nabla \mathcal{L}(r)\|^2
    +\lambda_{\max}\tfrac{\lambda_1 \min_v p[v]}{\lambda_{\max} \max_v p[v]}\max_v p[v]\|\nabla \mathcal{L}(r)\|\|\nabla \mathcal{L}(r)\| 
    \\&\quad= 0
\end{align*}
Therefore, 
the loss asymptotically converges to a locally optimal point.
\end{proof}

\subsubsection{Known robustness limits} Following a similar approach as to check whether graph filters are positive definite in Equation~\ref{condition}, we can also bound the eigenvalue ratio by knowing only the type of filter but not the graph or priors. In particular, for graph filters $H(W)$ where $W$ are symmetric adjacency matrices of undirected graphs:
\begin{equation}
    \frac{\lambda_1}{\lambda_{\max}}\geq \frac{\min_{\lambda \in[-1,1]} H(\lambda)}{\max_{\lambda\in[-1,1]}H(\lambda)}
\end{equation}

Since personalized pagerank can be expressed in closed form as the filter $(1-a)(I-aW)^{-1}$ and heat kernels as $e^{-t(I-W)}$, where $a$ and $t$ are their parameters, their respective eigenvalue ratios for $\lambda\in[-1,1]$ are at most $\tfrac{1-a}{1+a}$ and $e^{-2t}$. For larger values of $a$ and $t$, which penalize less the spread of graph signal priors farther away, stricter optimizers are required to keep track of the gradient's negative slope. When normalization is the only post-processing employed, all elements of the personalization vector are the same and bounds for sufficient optimizer strictness coincide with the aforementioned eigenvalue ratios. In practice, these bounds are often lax compared to the small average node score posteriors $\tfrac{1}{|\mathcal{V}|}$ arising from L1 normalization in graphs with many (such as thousands or millions of) nodes. For example, for personalized pagerank with $a=0.85$ it suffices to select $0.081$-optimizers to edit priors. Even for wider prior diffusion with $a=0.99$ it suffices to select $0.005$-optimizers.

\subsubsection{Prior edits as optimizers} Based on the above analysis, we finally explain why, if the surrogate model of bias estimation matches real-world behavior, the proposed personalization error-based editing mechanism of Equation~\ref{FairEdit} can tweak posteriors to locally optimize the set fairness-aware objectives. To do this, in Theorem~\ref{anyedit} we translate the required tightness of optimizers to a corresponding tightness of projecting loss gradients to vector spaces of prior editing model parameter gradients. When this requirement is met, there exist prior editing model parameters that lead posteriors to locally optimize the objective. Since the same quantity as in Theorem~\ref{robustness} is used to decide adequate tightness, the analysis of the previous subsection indicates that even coarse surrogate models can be met with success.
\par
Intuitively, in Theorem~\ref{anyedit} the matrix $E$ computes the difference between the unit matrix and multiplying the matrix of parameter gradients $D_F(\theta)$ with its right pseudo-inverse; these differences are further constrained to matter only for nodes with high gradient values. This means that, as long as the objective's gradient has a clear direction to move towards to and this can be captured by the surrogate prior editing model $F(\theta)$, a parameter trajectory path exists to arrive at locally optimal prior edits. For example, if prior editing had the same number of parameters as the number of nodes and its parameter gradients were linearly independent, $D_F(\theta)$ would be a square invertible matrix, yielding $D_F(\theta)(D_F^T(\theta)D_F(\theta))^{-1}D_F^T(\theta)=\mathcal{I}\Leftrightarrow E=\textbf{0}$ and the theorem's precondition inequality would always hold. Whereas as the number of parameters decreases, it becomes more important for $F(\theta)$ to be able to induce degrees of freedom in the same directions as its induced loss's gradients, which our theoretical analysis aims to approximate with the self-consistency criterion that prior edits should induce fairness-aware posteriors.
\par
At this point we stress that this section's theoretical analysis indicates that prior editing models with few parameters \textit{could} yield locally optimal priors, since it suffices to form only loose approximations of desired properties. On the other hand, due to the mathematical intractability of plugging in true prior editing gradients in the computation of the pseudo-inverse, the real-world efficacy of proposed prior editing schemes needs to be experimentally corroborated, as we do in the next section. As a final remark, we stress that being able to explicitly retain priors for some parameters, as the new proposed model of Equation~\ref{FairEdit} (but not the previous model of Equation~\ref{FP}) does, is a necessary condition for our analysis to hold true.

\begin{theorem}\label{anyedit}
Let us consider graph signal priors $q_0$, positive definite graph filter $H(W)$ with largest and smallest eigenvalues $\lambda_{\max},\lambda_1$, post-processing vector $p$, differentiable loss function $\mathcal{L}(r)$ in domain $\mathcal{R}$ and a differentiable graph signal editing function $F(\theta)$ with parameters $\theta\in \mathbb{R}^K$ for which there exist parameters $\theta_0$ satisfying $F(\theta_0)=q_0$ and $pH(W)F(\theta)\in\mathcal{R}$. Let us consider the $|\mathcal{V}|\times K$ table function $D_F(\theta)$ with rows:$$D_F[v]=\nabla_\theta (F(\theta)[u])$$
where $\nabla_\theta$ indicates parameter gradient vectors. If for any parameters $\theta$ it holds that:
\begin{align*}
    &\big\|E\nabla \mathcal{L}(r)\big\|< \tfrac{\lambda_1 \min_v p[v]}{\lambda_{\max} \max_v p[v]}\|\nabla \mathcal{L}(r)\|\quad\text{ for }\nabla \mathcal{L}(r)\neq \textbf{0}\\
    &E = \mathcal{I}-D_F(\theta)(D_F^T(\theta)D_F(\theta))^{-1}D_F^T(\theta)\\
    &r = diag(p)H(W)F(\theta)
\end{align*}
then there exist parameters $\theta_{\infty}$ that make respective posteriors $r_\infty=diag(p)H(W)F(\theta_{\infty})$ be locally optimal.
\end{theorem}

\begin{proof}
Let us consider a differentiable trajectory for parameters $\theta(t)$ for times $t\in[0,\infty)$ that starts from $\theta(0)=\theta_0$ and (asymptotically) arrives at $\theta_\infty=\lim_{t\to\infty} \theta(t)$. For this trajectory, it holds that $\tfrac{d F(\theta(t))}{dt}=D_F(\theta)\frac{d\theta}{dt}$. Then, let us consider the posteriors $r(t)=diag(p)H(W)F(\theta(t))$ arising from the graph signal function $F(\theta(t))$ at times $t$, as well the least square problem of minimizing the projection of the loss's gradient to the row space of $D_F(\theta(t))$: $$\text{ minimize }\|\nabla \mathcal{L}(r(t))-D_F(\theta(t))x(t)\|$$
The closed form solution to this problem can be found by: $$x(t)=(D_F^T(\theta)D_F(\theta))^{-1}D_F^T(\theta)\nabla \mathcal{L}(r(t))$$
Thus, as long as $q(t)=F(\theta(t))$ (Proposition I), the theorem's precondition for posteriors $r(t)=diag(p)H(W)q(t)$ and $\nabla \mathcal{L}(r(r))\neq \textbf{0}$ can be written as:$$\|\nabla \mathcal{L}(r(t))-D_F(\theta(t))x(t)\|<  \tfrac{\lambda_1 \min_v p[v]}{\lambda_{\max} \max_v p[v]}\|\nabla \mathcal{L}(r(t))\|$$
Hence, if we set $x(t)$ as parameter derivatives:
$$\tfrac{d \theta(t)}{dt}=x(t)\Rightarrow \|\nabla \mathcal{L}(r(t))-\frac{d F(\theta(t))}{dt}\|<  \tfrac{\lambda_1 \min_v p[v]}{\lambda_{\max} \max_v p[v]}\|\nabla \mathcal{L}(r(t))\|$$
which means that, on the selected parameter trajectory, $\tfrac{F(\theta(t))}{dt}$ is a $\tfrac{\lambda_{\max} \max_v p[v]}{\lambda_1 \min_v p[v]}$-optimizer of the loss function. As, such, from Theorem~\ref{robustness} it discovers locally optimal posteriors.
\par
As a final step, we now investigate the priors signal editing converges at. To do this, we can see that the update rule leads to the selection of priors $q(t)$ at times $t$ for which:
\begin{align*}
    &\tfrac{d q(t)}{dt}=\tfrac{d F(\theta(t))}{dt}
    \\&\Rightarrow \lim_{t\to\infty} q(t)=q(0)+\int_{t=0}^{\infty} \tfrac{F(\theta(t))}{dt} dt = F(\theta(t_{\infty}))- F(\theta(0))=F(\theta(t_{\infty}))
\end{align*}
\par
We can similarly show that (Proposition I) holds true. Hence, there exists an optimization path of prior editing parameters (not necessarily the same as the one followed by the optimization algorithm used in practice) that arrives at the edited priors $F(\theta(t_{\infty}))$ for some parameters $\theta(t_{\infty})$ that let posteriors exhibit local optimality.
\end{proof}

\section{Experiment Setup}
In this section we describe the experiment settings (graphs and base graph filters), competing approaches and evaluation methodology used to assess whether prior editing can make graph filter posteriors fair.
\subsection{Graphs}\label{graphs}
Our experiments are conducted on 12 real-world graphs from the domains of social networking, scientific collaboration and software engineering. The graphs are retrieved from publicly available sources, namely the SNAP repository of large networks \citep{snapnets}, the network repository \citep{nr}, the LINQS repository \citep{linqs}, the AMiner repository of citation data sets \citep{tang2008arnetminer} and the software dependency graph data set \citep{Musco2016}. They are selected on merit of seeing widespread use in their respective domain research and comprising multiple node attributes, which we use as predictive and sensitive information.
\par
In all graphs, nodes are organized into potentially overlapping communities based on their attributes. Motivated by the frequent use of graph filters for recommendation, we consider membership to each community as a different binary attribute. When not stated otherwise, and to facilitate experiments with a small portion of seed nodes, we select the first community with more than $100$ nodes to experiment on. If there are no explicit sensitive node attributes, we designate members of the second community found during the previous search as sensitive.
\par
The sensitive attributes of social networking graphs could give rise to discriminative behavior against persons, such as recommending other attributes with disproportionately lower scores. Disparate impact mitigation safeguards against this kind of disparity. This is also useful in other types of graphs, where biases are less impactful from a humanitarian perspective, but can still impact the inclusiveness of recommendation to end-users. For example, in  scientific collaboration graphs, it could be important for older publishing venues to compete fairly with newer ones by not biasing recommendation scores against them, for example due to fewer handled publications. Or it could be important to make searches of scientific publications respect a lesser-known area of research, such as less known types of diabetes in the PubMed graph below. Finally, in software engineering graphs it is often useful to discover entities (such as libraries, source code artifacts) related to query ones while protecting those of a subsystem from being avoided, for example because it plays a pivotal role in a software project's architecture and potential risks of impacting it should be well-understood.
\par
Below we detail the type of data captured by each graph. Quantifiable characteristics, such as the number of graph nodes, edges, positive labels, sensitive labels and pRule of positive labels given sensitive ones are summarized in Table~\ref{tab:data sets}. For most graphs, the pRule of positive labels $0$, which indicates that either all positive labels are sensitive or all of them are not sensitive (flipping which nodes are considered sensitive retains data set pRule and experiment results).
\paragraph{ACM \citep{tang2008arnetminer}} A co-authorship graph whose nodes are authors forming edges based on whether they have co-authored a publication. We consider communities corresponding to different publication venues (such as journals or conferences) the authors have published in. To construct this data set, we processed the authorship and publication data of the 2017 version of the ACM citation data set extracted by AMiner.\footnote{\footnotesize ACM-Citation-network V9 from \url{https://aminer.org/citation}}
\paragraph{Amazon \citep{leskovec2007dynamics}} A graph comprising frequent Amazon product co-purchases, as well as their category. The graph was parsed from frequent co-purchase metadata hosted in the SNAP repository.\footnote{\url{https://snap.stanford.edu/data/amazon-meta.html}}
\paragraph{Ant \citep{Musco2016}} A method call graph for the Apache Ant dependency builder tool, where nodes are source code methods and edges indicate that one of the methods called the other. We retrieve this graph from the feature file corresponding to the project's release 1.9.2 in the software dependency graph data set and remove dangling nodes with only one edge. We consider methods of the same class, as identified through their signatures, to belong to the same community. 
Hence, predicting communities corresponds to predicting the organization of methods into classes.
\paragraph{Citeseer \citep{getoor2005link}} A citation graph where scientific publications are nodes and edges correspond to citations between them. Nodes are categorized into communities based. We use the version of the data set provided by the LINQS repository. Since our aim is to experiment on attribute propagation mechanisms, we do not experiment with available publication text tokens that could be used by more complex schemes (such as predict-then-propagate graph neural networks \citep{klicpera2018predict} that estimate graph signals through multilayer perceptrons before propagating them).
\paragraph{DBLP \citep{tang2008arnetminer}} A co-authorship graph whose nodes are authors forming edges based on whether they have co-authored a publication. We consider communities corresponding to different publication venues the authors have published in. To construct this data set, we processed the authorship and publication data of the 2011 version of the DBLP citation data set extracted by AMiner\footnote{\footnotesize DBLP-Citation-network V4 from \url{https://aminer.org/citation}}.
\paragraph{Facebook0 \citep{leskovec2012learning}} A Facebook graph of user friendships starting from given user and record social relations between them and their friends, including relations between friends. Ten such graphs are available in the source material, out of which we experiment on the first one. We use the version of the data set hosted by SNAP and select the anonymized binary `gender' attribute as sensitive and the first anonymized binary `education' attribute as the prediction label.
\paragraph{Facebook686 \citep{leskovec2012learning}} A graph obtained through the same process as Facebook0 and choosing a different graph out of those available in the source material (the ego network of user with identifier 686).
\paragraph{Log4J \citep{Musco2016}} A method call graph for the Java logging project Log4J, where nodes are source code methods and edges indicate that one of the methods called the other. We retrieve this graph from the feature file corresponding to the project's release 2.0b9 in the software dependency graph data set and apply the same preprocessing and community extraction steps as we did for Ant. 
\paragraph{Maven \citep{benelallam2019maven}} A dependency graph of Java libraries hosted in Maven central.\footnote{\url{https://zenodo.org/record/1489120\#.YCnNumgzaUk}} Nodes correspond to libraries and edges to dependencies between them, i.e. that one of the edge's nodes depends on the other. We organize libraries into communities based on their parent projects (for example, the libraries \textit{org.seleniumhq.selenium:selenium-server:3.0.1} and \textit{org.seleniumhq.selenium:selenium-support:3.0.1} are both consider part of the \textit{org.seleniumhq.selenium} project) and remove dangling nodes with only one edge.
\paragraph{Pubmed \citep{namata2012query}} A citation data set of PubMed publications pertaining to diabetes research, where nodes correspond to papers and edges to citations between them. Nodes form communities based on the type of diabetes they research. We use the version of the data set provided by the LINQS repository. 
\paragraph{Squirrel \citep{Musco2016}} A method call graph for the Squirrel email server written in Java, where nodes are source code methods and edges indicate that one of the methods called the other. We retrieve this graph from the feature file corresponding to the project's release 0.34 in the software dependency graph data set and apply the same preprocessing and community extraction steps as we did for Ant and Log4J. 
\paragraph{Twitter \citep{nr}} A Twitter graph of political retweets, where nodes are social media users and edges indicate which users have retweeted posts of others. This data set comprises only one anonymized attribute of binary political opinions (left or right). We consider this attribute as sensitive and its complement as prediction labels.

\begin{table}[tbp]
\footnotesize
\centering
    \begin{tabular}{l r r r r r}
         \textbf{data set} & \textbf{Nodes} & \textbf{Edges} & \textbf{pRule} & \textbf{Positive} & \textbf{Sensitive} \\
         \hline
         ACM & 505,016 & 1,137,011 & 51\% & 52,222 & 1,098\\
         Amazon & 554,789 & 1,545,228 & 0 & 280,507 & 64,915\\
         Ant & 4,920 & 9,393 & 0 & 783 & 336\\
         Citeseer & 3,327 & 4,676 & 0 & 596 & 668\\
         DBLP & 978,488 & 3,491,030 & 10\% & 1,039 & 193\\
         Facebook0 & 333 & 2519 & 91\% & 225 & 120\\
         Facebook686 & 170 & 3,312 & 91\% & 94 & 78 \\ 
         Log4J & 2,121 & 3,600 & 0 & 115 & 272\\
         Maven & 76,137 & 425,339 & 0 & 185 & 488\\
         Pubmed & 19,717 & 44,327 & 0 & 4,103 & 7,739\\
         Squirrel & 6,791 & 10,880 & 0 & 232 & 372\\
         Twitter & 18,470 & 48,365 & 0 & 11,355 & 7,115\\
    \end{tabular}
    \caption{Characteristics of data sets involved in experiments.}
    \label{tab:data sets}
\end{table}

\subsection{Base Graph Filters}\label{base filters}
For our experiments we consider the two popular types of graph filters we outline in Subsection~\ref{gsp}: personalized pagerank and heat kernels. Commonly used parameters for those filters that encourage few-hop propagation are $a=0.85$ and $t=3$. However, previous findings suggest that propagating graph signal priors  more hops away leads to higher posterior AUC for communities with many nodes \citep{krasanakis2020boosted}. We thus experiment with propagation parameters $a=0.99$ and $t=5$ that induce this behavior too. Since the sweep procedure is also a well-known method in improving posterior score quality (for example, it often increases AUC) we also experiment both with and without it. In total, depending on the choice of filter type, propagation parameter and post-processing, we experiment across $2 \cdot 2 \cdot 2=8$ base graph filters, which we outline in Table~\ref{tab:base filters}.

\begin{table}[tbp]
\footnotesize
\centering
    \begin{tabular}{l c c c c}
         \textbf{Annotation} & \textbf{Type of Filter} & \textbf{Parameter} & \textbf{Sweep Ratio Post-processing} \\
         \hline
         PPR.85 & personalized pagerank & $a=0.85$ & \\
         PPR.99 & personalized pagerank & $a=0.99$ & \\
         HK3 & heat kernels & $k=3$ & \\
         HK7 & heat kernels & $k=7$ & \\
         PPR.85S & personalized pagerank & $a=0.85$ & \checkmark  \\
         PPR.99S & personalized pagerank & $a=0.99$ & \checkmark \\
         HK3S & heat kernels & $k=3$ & \checkmark \\
         HK7S & heat kernels & $k=7$ & \checkmark \\
    \end{tabular}
    \caption{Base graph filters used in experiments.}
    \label{tab:base filters}
\end{table}
\par
In all cases, we treat graphs as undirected ones and employ symmetric normalization. This way, graph filters become positive definite and hence support our theoretical results. Posterior scores are computed to numerical precision of $10^{-9}$ using the implementations of chosen graph filters provided by the \textit{pygrank}\footnote{\url{https://pypi.org/project/pygrank/}} graph ranking library.

\subsection{Compared Approaches}\label{compared approaches}
In this subsection we outline promising existing and new approaches that improve the fairness of base graph filters. These span multiple fairness-aware objectives, such as maximizing fairness, meeting fairness constraints and trying to preserve the posterior quality measured by AUC, which we detail in Table~\ref{tab:methods}. Their implementation has been integrated in the \textit{pygrank} library.
\paragraph{None} The base graph filter.
\paragraph{Mult} A simple post-processing baseline that multiplies posterior scores across the sensitive and non-sensitive groups with a different constant each, so that disparate impact is fully mitigated. If $r$ are the base graph filter's posteriors and $s$ a graph signal holding the sensitive attribute, this method post-processes posteriors per the rules:
$$r_{Mult}[v]=\bigg(\frac{\phi s[v]}{\sum_{u\in S}s[u]r[u]}+\frac{(1-\phi)(1-s[v])}{\sum_{u\not\in S}s[u]r[u]}\bigg)r[v]$$
where $\phi=\frac{|S|}{|S|+|S'|}$ is the fraction of graph nodes that are sensitive and $s[u]=\{1\text{ if }u\in S, 0\text{ otherwise}\}$. It holds that $\sum_{v\in S}r_{Mult}[v]=\sum_{v\not\in S}r_{Mult}[v]$.
\paragraph{FairWalk} A random walk strategy previously used for fair node embeddings. We implement this as an asymmetric preprocessing applied to the graph's adjacency matrix that retains nodes degree but makes signals spread equally between sensitive and non-sensitive neighbors. This approach aims to predominantly maintain the propagation mechanism and attempts to make posteriors fair only if this is convenient by encountering nodes with both sensitive and non-sensitive neighbors.
\paragraph{LFPRO} Near-optimal redistribution of ranks causing disparate impact. Contrary to our assumption that aims to influence posteriors by editing priors, this approach directly offsets the disparate impact of posteriors by moving excess node scores between the sensitive and the non-sensitive nodes to improve the pRule as much as possible while maintaining non-negative scores. To avoid numerical underflows that erroneously prevent this approach from exact convergence in the graphs with many nodes, we repeat the gradual movement of scores up to numerical tolerance $10^{-12}$.
\paragraph{FairPers} Our previously proposed fair personalization model described by Equation~\ref{FP}. Its tradeoff parameters assume values $a_S,a_{S'}\in[0,1]$, whereas we limit exponentials to large enough range $b_S,b_{S'}\in[-10,10]$ that, if needed, lets some posteriors dominate the editing mechanism for large exponents of vanish for small exponents. This variation aims to maximize fairness while partially retaining posterior quality and is hence trained towards optimizing Equation~\ref{objective} for $w_{prule}=1$ and $sup_{prule}=1$. Parameters are tuned with an algorithm we developed as part of the \textit{pygrank} library for non-derivative optimization, which is detailed in Appendix~\ref{tuning}.
\paragraph{FairPers-C} A variation of FairPers that aims to impose strong fairness constraints with $w_{prule}=10$ of optimizing the pRule up to value $sup_{prule}=80\%$.
\paragraph{FairEdit} The personalization editing mechanism of Equation~\ref{FairEdit} proposed in this work. Its parameters assume values $\{\theta_K\in[-1,1]|K=0,1\}$ and $a_{pers}\in[0,1]$ and this variation trains them to maximize fairness while partially retaining posterior quality by optimizing Equation~\ref{objectiveKL} for $w_{pRule}=1$ and $sup_{pRule}=1$. Parameters are tuned with the same optimization algorithm as FairPers.
\paragraph{FairEdit-C} A variation of FairEdit that aims to impose strong fairness constraints with $w_{prule}=10$ of optimizing the pRule up to value $sup_{prule}=80\%$.

\begin{table}[tbp]
\footnotesize
\centering
    \begin{tabular}{l l c c c c}
         ~& & \textbf{Preserve} & \textbf{Maximimize} & \textbf{Fairness} \\
         \textbf{Method} & & \textbf{Posteriors} & \textbf{Fairness} & \textbf{Constraints} \\
         \hline
         None & \\
         Mult & (baseline) & & \checkmark &\\
         Fairwalk & \citep{rahman2019fairwalk} & \checkmark & Partially &\\
         LFPRO & \citep{tsioutsiouliklis2020fairness} & Partially & \checkmark &\\
         FairPers & \citep{krasanakis2020applying} & Partially & \checkmark & \checkmark\\
         FairPers-C & \citep{krasanakis2020applying} & \checkmark & & \checkmark\\
         FairEdit & (this work) & Partially & \checkmark & \checkmark\\
         FairEdit-C & (this work) & \checkmark &  & \checkmark\\
    \end{tabular}
    \caption{Explicit objectives of fairness-aware methods.}
    \label{tab:methods}
\end{table}

\subsection{Evaluation Methodology}\label{methodology}
This subsection details the methodology we follow to assess fairness-aware approaches on combinations of graphs and base graph filters, as well as the summary statistics used to perform a high level comparison between approaches.
\subsubsection{Train-test splits}
To rigorously evaluate fairness-aware approaches, we separate graph nodes into training and test sets, by uniformly sampling the former without repetition to comprise one of the fractions $\{10\%,20\%,30\%\}$ of all graph nodes. Positive labels are on average also split alongside the same fraction to construct graph signal priors (i.e. by assigning $1$ to their respective elements). For example, for the Maven data set a $20\%$ train-test split uses only $20\%\cdot 185=37$ of positive labels to construct graph signal priors. For each graph, we experiment with all three split ratios and average evaluation measures between them. To make sure that comparisons between different approaches are not affected by sampling variance, we use seeded sampling to select training nodes.
\subsubsection{Measures} Our experiments are focused on two types of fairness-aware objectives: maximizing the pRule of posterior scores and achieving high pRule values while preserving high posterior score quality. Both consider a recommender system setting, where we start from graph signal priors and graph filters are used to obtain node score posteriors. These posteriors then provide a granular understanding on how much nodes pertain to an attribute shared by nodes with non-zero priors. To assess how well approaches achieve high posterior quality and fairness we respectively calculate the AUC and pRule of their posteriors over the evaluation subset of nodes.
\par
An overview of information flow during evaluation is presented in Figure~\ref{fig:evaluation}. It is worth noting that, although only test nodes are used to calculate measures, the sensitive-aware approaches explored in this work require knowledge of sensitive attributes for all nodes, which include the test ones. Nonetheless, evaluation remains robust in the sense that training data are not overfitted as long as fairness-aware approaches are not provided with information about \textit{which} nodes are used for testing.\footnote{\footnotesize In our previous work, we performed robust fairness evaluation by keeping the minimum pRule between considering all vs considering only test nodes. However, in practice, we found that the pRule of test nodes dominates this quantity.} For example, approaches such as \textit{Mult} induce perfect fairness when considering all nodes, but this fairness needs to generalize to subsets of graph nodes, such as test ones.
\begin{figure}[b!]
\centering
    \includegraphics[width=1\textwidth,clip]{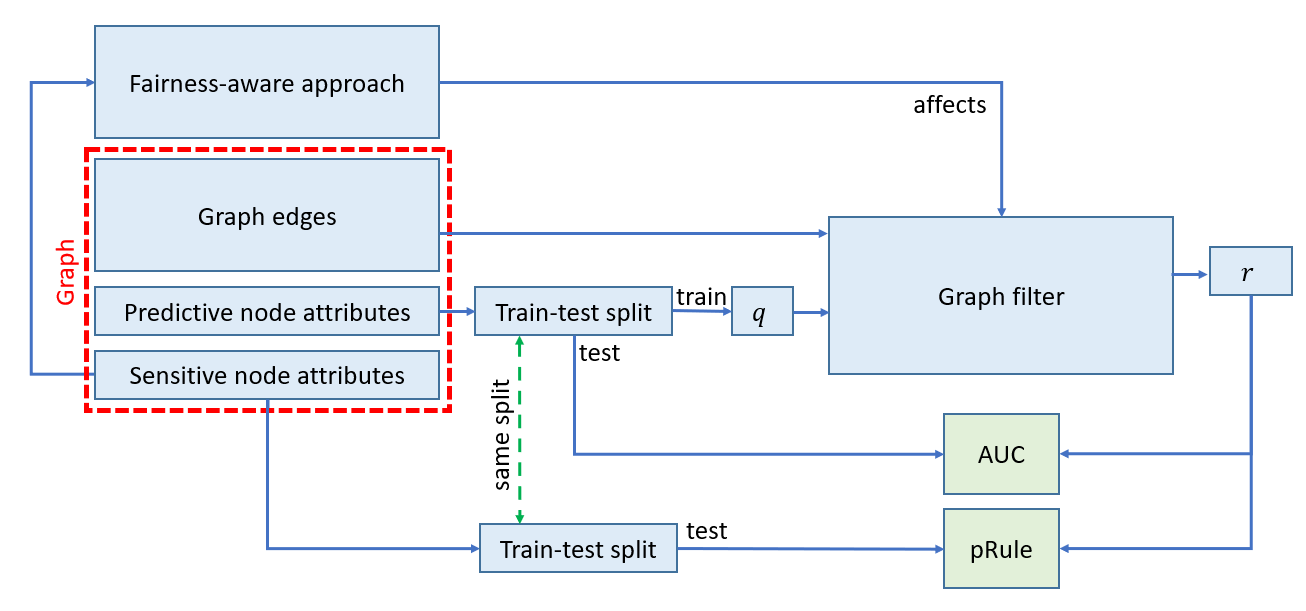}
    \caption{The process of assessing the AUC and pRule of a fairness-aware approach on a graph filter for a given graph and train-test split ratio.}
    \label{fig:evaluation}
\end{figure}
\subsubsection{Approach comparison}
Our experiments output one value for each measure per combination of graph, data set, fairness-aware approach and graph filter, for a total of $12\cdot 8\cdot 8=768$ combinations. To help gather insights from these results, we opt to delegate them to Appendix~\ref{results} and instead extract high-level summary statistics that make it easy to compare different approaches.
\par
To this end, we extract three summary statistics across all graphs and graph filters, depending on the type of post-processing (no post-processing or the sweep ratio). These statistics are: a) the average AUC of approaches and their statistical significant differences, b) the average pRule of approaches and their statistical significant differences and c) the percentage of experiments in which pRule achieves the constraint of $80\%$ or more. As an example, in Figure~\ref{fig:average} we demonstrate the process of obtaining the average AUC and pRule values for the {Mult} approach, where the same procedure is followed for all others too.
\begin{figure}[tbp]
    \centering
    \includegraphics[width=1\textwidth,clip]{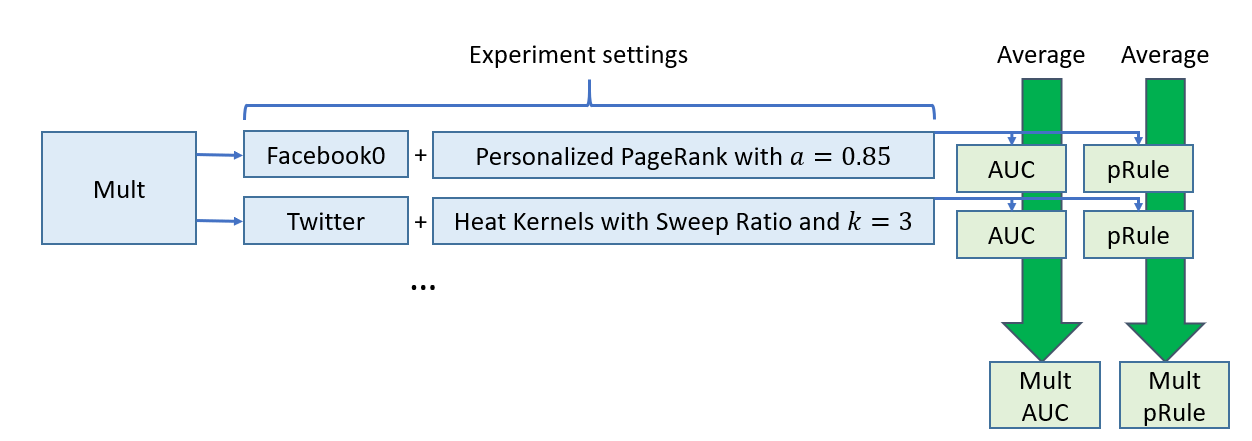}
    \caption{Averaging the outcome of \textit{Mult} AUC and pRule values across experiment settings.}
    \label{fig:average}
\end{figure}

\par
We assert which approaches differ significantly from the rest by following a well-known procedure for multi-approach comparison \citep{demvsar2006statistical,garcia2008extension,derrac2011practical}. In detail, we first employ a Friedman test with p-value $<$0.001 to assert that at least one approach differs from the others with statistical significance. Provided that the Friedman test rejects the null hypothesis that all approaches produce similar outcomes, we then use a Nemenyi post-hoc test to compare individual approaches. The latter ranks approaches for each measure (assigning rank $1$ to the best approach, $2$ to the second best and so on), reports the average rank across data sets and base ranking algorithms and outputs a \textit{critical difference} with p-value at most $0.05$, where average rank differences greater than it imply statistical significance at that p-value level.\footnote{\footnotesize The Nemenyi test is weaker than most other post-hoc tests, such as the Holm test, in that it does not necessarily discover all of statistically significant differences. But we prefer it on merit that it reports a clear order between approaches and a universal criterion to determine statistically significant differences. The p-value for this test is not selected to be too tight to offset its weakness.}

\section{Experiment Results}
Following the evaluation methodology detailed in the previous section, we compare the fairness-aware approaches proposed in this work ({FairEdit}, {FairEdic-C}) to the ones of our previous paper ({FairPers}, {FairPers-C}), an intuitive baseline ({Mult}) and two similar methods employed in the literature ({FairWalk}, {LFPRO}). This comparison involves running experiments for all combinations of graph filters and data sets and exploring the resulting AUC and pRule values presented in Appendix~\ref{results}.
\par
Table~\ref{tab:comparison} forms a high-level summary of experiment results. In particular, it reports the average of evaluation measures across all experiments for both types of post-processing. It also presents the Nemenyi ranks for multiway comparison between approaches. We remind that lower ranks indicate approaches performing better for the particular measure, where rank differences exceeding a critical difference indicate statistically significant improvements. For our both sets of comparisons, the critical difference indicating statistical significance is calculated as approximately $1.6$. For example, base graph filters (the approach dubbed {None}) with no post-processing are assigned 3.0 AUC rank, which indicates that they outperform {LFPRO} in this regard with statistical significance, as the latter's corresponding rank is $5.4>3.0+1.6$.  All numbers rounded to their two most important digits.
\par
\begin{table}[tbp]
\footnotesize
\centering
    \begin{tabular}{l | r r r | r r r}
        & \multicolumn{3}{c|}{\textbf{No Post-processing}} & \multicolumn{3}{c}{\textbf{Sweep Ratio}}\\
          & \textbf{AUC} & \textbf{pRule} & \textbf{pRule $\geq$ 80\%} & \textbf{AUC} & \textbf{pRule} & \textbf{pRule $\geq$ 80\%} \\
          \hline
         None & .76 (3.0) & .52 (6.5) & .21 & .77 (3.2) & .51 (6.8) & .23\\
         Mult & .73 (4.3) & .68 (5.7) & .33 & .73 (5.0) & .68 (5.8) & .35\\
         LFPRO & .67 (5.4) & .83 (4.2) & .75 & .67 (6.0) & .82 (4.7) & .69\\
         FairWalk & .75 (3.3) & .47 (6.8) & .17 & .75 (3.9) & .54 (6.3)& .31\\
         FairPers & .66 (5.8) & .93 (2.2) & .92 & .66 (5.3) & .95 (2.2) & .96\\
         FairPers-C & .72 (4.6) & .86 (4.3) & .79 & .72 (4.7) & .88 (4.1) & .92\\
         FairEdit & .69 (4.6) & .92 (2.2) & .88 & .70 (4.6) & .94 (2.0) & .92\\
         FairEdit-C & .72 (4.1) & .86 (4.2) & .90 & .76 (3.3) & .88 (4.2) & .94\\
    \end{tabular} 
    \caption{Average AUC, pRule and fraction of experiments achieving $80\%$ pRule for fairness-aware approaches (higher are better). Average Nemenyi ranks for measures in parenthesis (smaller are better), where rank differences greater than 1.6 are statistically significant. Different results are presented depending on the post-processeing of the base graph filter.}
    \label{tab:comparison}
\end{table}

\subsection{Fairness maximization} Compared to base graph filters, the approaches LFPRO, FairPers and FairEdit reduce the average AUC by at least 7\% with statistical significance, where LFPRO suffers from the greatest posterior quality reduction. Posterior quality loss for all three methods arises from their attempt to maximize fairness while only partially preserving posteriors.
\par
Among these approaches, FairEdit exhibits at least 3\% greater average AUC compared to the second-best FairPers at the cost of $1\%$ pRule. Although these differences are not statistically significant, they nonetheless identify FairEdit as the preferred method to eliminate disparate impact in real-world applications. Looking at detailed experiment results, FairPers and FairEdit outperform each other in different experiments, which suggests that there is added value in better controlling prior retention in future research, which we find in Appendix~\ref{ablation study} to contribute the most in this improvement.
\par
In terms of disparate impact elimination, LFPRO yields significantly smaller pRule by 9\%  compared to the two prior editing approaches, which corroborates our assumption that respecting the graph's structure by editing priors can better satisfy posterior objectives than directly editing those through uniformly skewing mechanisms.

\subsection{Achieving fairness constraints} The Mult, FairPers-C and FairEdit-C approaches in large part preserve posterior score quality, with their AUC seeing a reduction by at most by 5\% compared to base graph filters, while improving their pRule by large margins. This achievement can be attributed to prior editing approaches placing significance to pRule improvements only up to the target level, whereas Mult affects posterior node score comparisons only between sensitive nodes and the rest of the graph, which in most graphs at introduces only a small fraction of erroneously high node posterior scores compared to the number of nodes with non-positive labels.
\par
Overall, prior editing approaches achieve over 80\% pRule both on average and a remarkably large portion of experiments, which hints at their robust generalization capabilities. By contrast, Mult yields significantly lower pRule on average with statistical significance regardless of the and does not meet the fairness criterion in most experiments. In fact, the constrained approaches exhibit similar -and on average higher- pRule than even LFPRO, while at the same time enjoying sigificantly higher AUC values.

\subsection{Sweep ratio improvements}  We previously mentioned that the sweep ratio is often employed to improve posterior quality for graph filters. In line with this consensus, the average AUC of approaches when this kind of post-processing is used is at least equal to or greater than the one with no post-processing. Although this improvement is often marginal ($\sim$1\%), in the case of FairEdit-C the sweep ratio improves AUC by $4\%$. More importantly, this mechanism induces 2\% pRule improvement for all personalization editing approaches, which suggests that there is likely merit in employing it in new settings.
Intuitively, improvements thanks to the sweep ratio can be attributed to increasing the posteriors of low-scored nodes so that it becomes easier to move a portion of those to sensitive ones without affecting pairwise node score comparisons.

\subsection{Experiment outliers} Of the fairness-aware approaches, FairWalk performs similarly or worse than base graph filters. However, result details reveal that, in some experiments, it manages to improve the pRule. Hence, we argue that this method's failure lies more in sensitive nodes not being randomly mixed in all graphs we experiment on but forming structurally close communities; this prevents the fair random walk formulation from visiting sensitive and non-sensitive neighbors with equal probabilities, as for most nodes only one of those kinds of neighbors is available.
\par
Lastly, we overview some noticeable outliers in experiment outcomes. First, Mult and LFPRO sometimes outperform our approaches in preserving AUC (such as the near-100\% AUC of the Maven data set). Our understanding of these phenomena lies in these two approaches potentially affecting only the sensitive nodes without respecting propagation mechanisms, which lets them find success in specific cases. 
Another case that warrants attention is the inability of our approach to meet 80\% pRule for the DBLP graph when base graph filters involve personalized pagerank. We attribute this finding to further approximation terms being needed given that graph's eigenvalue ratio, but further investigation is needed.

\section{Conclusions and Future Work}
In this work we explored the concept of editing graph signal priors to locally optimize fairness-aware graph filter posterior objectives. To this end, we proposed an editing mechanisms with few parameters and showed that, given its parameter gradients' lose approximation of fairness-aware objective optimization slopes, it reaches local optima without overfitting posterior node scores. We then experimented on a large number of real-world graphs with various graph filters and signal priors, where we used our approach to produce high-quality posterior scores in terms of AUC that are fair (enough). Our findings suggest that, compared to other fairness-aware methods, prior editing is able to better reduce disparate impact or meet reduction constraints while in large part preserving the predictive quality of posteriors.
\par
Promising applications of this work could involve bringing fairness to decoupled graph neural network architectures, where feature extraction and propagation with graph filters are performed separately. Our theoretical framework on optimizing posteriors is also general enough to warrant application on other types of objectives, such as maximizing smooth relaxations of AUC. Finally, we are interested in further exploring the assumption that prior editing in its current form exhibits enough degrees of freedom and aim to provide a bound on its necessary number of terms in future work.

\acks{This work was partially funded by the European Commission under contract numbers H2020-860630 NoBIAS and H2020-825585 HELIOS.}


\newpage
\appendix

\section{Parameter Tuning}\label{tuning}
In this appendix we detail the algorithm we developed as part of the \textit{pygrank} library to optimize the parameters of graph signal prior editing mechanisms (FairPers, FairEdit and their constrained variations) in a given hypercube without derivating posterior objectives.
\par
Our algorithm involves cycling through a list of parameters $\theta$ and optimizing a loss function $\ell(\theta)$ by finding the best permutation around each one. For the prior editing approach of this work, the parameters are $\theta=\{a_S,a_{S'},b_S,b_{S'}\}$ and evaluating the loss function $\ell(\theta)=\mathcal{L}(H(W)q_{est})$ once for the respective prior estimation $q_{est}$ involves running a graph filter; for the graphs with millions of nodes, this requires approximately $1-2$ sec on the 2.2 Ghz 4-core CPU with hyperthreading used in experiments, after a preprocessing operation obtains a common sparse version of the adjacency matrix (sparse matrices require $O(|\mathcal{E}|)$ amortized time in big-o notation to perform one-hop propagations by multiplying graph signals). Though this time is not prohibitively large, to run the large number of experiments in this work we design the optimization algorithm with the goal of performing only a small number of parameter loss function evaluations. Intuitively, it is equivalent to moving the center of the selected rectangle chosen for each parameter based on subsequent selections of other parameters. So, at the time where the parameter's permutation breadth shrinks to the size of its intended rectangle, potential combinations with permutations of other parameters have been considered too. 
\par
The implemented algorithm is a variation of divided rectangles (DIRECT) \citep{finkel2006additive} that, instead of keeping many candidate rectangles to divide, keeps only one, though of larger width than the partition. This practice corresponds to the shrinking radius technique proposed for non-convex block coordinate optimization \citep{lyu2020convergence}, although the two are not mathematically equivalent due to the finite sum of rectangle widths that limits the optimization within the hypercube of searched parameters,

\begin{algorithm}[b!]
\KwIn{parameter loss $\ell(\theta)$, parameter square bound vectors $\theta_{\max},\theta_{\min}\in\mathbb{R}^K$, tolerance $\epsilon$, line partitions $part$, range contraction $T$}
\KwOut{near-optimal vector of $K$ parameters}
$\theta\leftarrow (\theta_{\max}+\theta_{\min})/2$\\
$\Delta\theta\leftarrow \theta_{\max}-\theta_{\min}$\\
$err\leftarrow [\infty]\times K$\\
$i\leftarrow 0$\\
\While{$\max_i err[i]>\epsilon$}{
    $u_i\leftarrow $unit vector with $1$ at element $i$, $0$ elsewhere\\
    $\Theta\leftarrow \{\theta+u_i\cdot\Delta\theta[i]\cdot(p/part-1)\,|\, p=0,1,\dots,2\, part\}$\\
    $\theta\leftarrow \arg\min_{\theta\in\Theta} \ell(\theta)$\\
    $err[i]\leftarrow \max_{\theta\in\Theta} \ell(\theta)-\min_{\theta\in\Theta} \ell(\theta)$\\
    $\Delta\theta[i]\leftarrow \Delta\theta[i]/T$\\
    $i\leftarrow (i+1)\, mod\, K$
}
\Return{$\theta$}
\caption{Parameter tuning}\label{coordinate descent}
\end{algorithm}
\par
In detail, we start from the center of parameter square bounds and cycle through parameters $i$. For each of those parameters, we consider the maximal range $\Delta\theta[i]$ in which to search for new solutions and partition it uniformly to $2\, part+1$ points. These form a set $\Theta$ of parameter perturbations, out of which we select the ones minimizing the loss. Finally, we contract the range in which to search parameters by dividing it with the value $K$ and move on to the next parameter. We stop cycling through parameters when the max loss difference between perturbations are smaller than the given tolerance. This process is outlined in Algorithm~\ref{coordinate descent}.

\par
If the objective $\ell(\theta)$ is Lipshitz continuous with Lipshitz constant $\sup \|\nabla\ell(\theta)\|<\infty$, it is easy to see that the division of the parameter permutation radius by $T$ every $K$ iterations lets the algorithm run in amortized time $O\big(K(\text{run }\ell(\theta))\log_T (\|\theta_{\max}-\theta_{\min}\|_\infty\sup \|\nabla\ell(\theta)\|)\big)$. For experiments in this work, and to reduce running time as much as possible, we empirically selected the parameters $part=2,T=2,\epsilon=0.01$, which yielded near-identical results to employing the more granular parameters $part=4,T=1.3,\epsilon=0.001$ in a randomly selected subset of $15$ experiment settings.

\section{Detailed Experiment Results}\label{results}
\begin{table}[htbp]
    \footnotesize
    \setlength\tabcolsep{.5pt}
    \begin{tabular}{l | r r | r r | r r | r r | r r | r r | r r | r r}
        ~&\multicolumn{2}{c|}{\textbf{None}}&\multicolumn{2}{c|}{\textbf{Mult}}&\multicolumn{2}{c|}{\textbf{LFPRO}}&\multicolumn{2}{c|}{\textbf{FairPers}}&\multicolumn{2}{c|}{\textbf{FairPers-C}}&\multicolumn{2}{c|}{\textbf{FairEdit}}&\multicolumn{2}{c|}{\textbf{FairEdit-C}}&\multicolumn{2}{c}{\textbf{FairWalk}}\\
        ~ & AUC & pRule & AUC & pRule & AUC & pRule & AUC & pRule & AUC & pRule & AUC & pRule & AUC & pRule & AUC & pRule\\
        \hline
        \multicolumn{15}{c}{Graph Filter PPR.85}\\
         \hline
ACM  & .71 & .48 & .71 & .74 & .74 & .92 & .73 & .95 & .76 & .95 & .75 & .81 & .74 & .92 & .71 & .81 \\
Amazon  & .81 & .01 & .79 & .82 & .59 & .93 & .75 & .99 & .73 & .88 & .68 & .98 & .72 & .83 & .81 & .03 \\
Ant  & .78 & .59 & .77 & .78 & .79 & .87 & .58 & .99 & .77 & .86 & .76 & .95 & .78 & .87 & .77 & .79 \\
Citeseer  & .84 & .13 & .83 & .74 & .77 & .73 & .69 & .91 & .82 & .83 & .87 & .43 & .82 & .82 & .84 & .16 \\
DBLP  & .92 & .75 & .92 & .77 & .92 & .86 & .86 & .60 & .86 & .60 & .89 & .66 & .89 & .66 & .92 & .09 \\
Facebook0  & .57 & .91 & .54 & .74 & .54 & .76 & .56 & .90 & .57 & .83 & .53 & .98 & .58 & .87  & .54 & .77 \\
Facebook686  & .52 & .94 & .51 & .74 & .50 & .70 & .48 & .93 & .52 & .90 & .50 & .99 & .52 & .90  & .53 & .95 \\
Log4J  & .78 & .77 & .77 & .75 & .76 & .83 & .62 & .99 & .68 & .89 & .74 & .94 & .76 & .92 & .77 & .68 \\
Maven  & 1.00 & .50 & 1.00 & .83 & 1.00 & .84 & .49 & .94 & .82 & .91 & .86 & .91 & .86 & .91 & 1.00 & .25 \\
Pubmed  & .90 & .45 & .82 & .77 & .81 & .97 & .64 & .99 & .78 & .83 & .86 & .69 & .79 & .86 & .88 & .65 \\
Squirrel  & .64 & .36 & .64 & .72 & .64 & .44 & .67 & .97 & .73 & .87 & .75 & .87 & .75 & .83 & .65 & .33 \\
Twitter  & .97 & .13 & .48 & .80 & .31 & .84 & .58 & .99 & .68 & .82 & .58 & .99 & .68 & .83 & .87 & .33 \\
    
        \hline
        \multicolumn{15}{c}{Graph Filter PPR.99}\\
         \hline
ACM  & .70 & .53 & .70 & .88 & .76 & .89 & .65 & .92 & .73 & .88 & .72 & .99 & .73 & .90 & .70 & .76 \\
Amazon  & .83 & .01 & .80 & .95 & .68 & .84 & .74 & .99 & .69 & .91 & .68 & 1.00 & .71 & .82 & .82 & .03 \\
Ant  & .71 & .92 & .70 & .94 & .70 & .96 & .64 & .99 & .68 & .94 & .54 & .99 & .57 & .96 & .72 & .76 \\
Citeseer  & .84 & .40 & .83 & .90 & .85 & .84 & .66 & .95 & .79 & .81 & .52 & 1.00 & .82 & .81 & .84 & .23 \\
DBLP  & .92 & .66 & .92 & .73 & .92 & .84 & .91 & .70 & .91 & .70 & .92 & .68 & .91 & .69 & .92 & .09 \\
Facebook0  & .56 & .93 & .56 & .96 & .56 & .96 & .58 & .97 & .59 & .94 & .53 & .97 & .47 & .88 & .54 & .77 \\
Facebook686  & .51 & .98 & .51 & .98 & .51 & .98 & .51 & .99 & .51 & .98 & .51 & .99 & .51 & .98 & .51 & .96  \\
Log4J  & .74 & .94 & .74 & .95 & .74 & .97 & .65 & .99 & .65 & .99 & .67 & .97 & .56 & .94 & .74 & .58 \\
Maven  & .98 & .77 & .98 & .99 & .98 & 1.00 & .46 & 1.00 & .93 & .80 & .87 & .83 & .89 & .80 & .98 & .21  \\
Pubmed  & .79 & .75 & .73 & .98 & .72 & .97 & .55 & 1.00 & .73 & .86 & .55 & 1.00 & .75 & .85 & .70 & .97 \\
Squirrel  & .63 & .46 & .63 & .84 & .63 & .59 & .64 & .96 & .70 & .72 & .71 & .90 & .70 & .74 & .63 & .52 \\
Twitter  & .74 & .69 & .57 & .98 & .39 & .98 & .59 & .99 & .66 & .85 & .59 & .99 & .69 & .81 & .73 & .15 \\
    
        \hline
        \multicolumn{15}{c}{Graph Filter HK3}\\
         \hline
ACM  & .68 & .49 & .68 & .41 & .66 & .91 & .72 & .99 & .73 & .94 & .71 & .96 & .73 & .93 & .68 & .60 \\
Amazon  & .88 & .01 & .87 & .49 & .46 & .84 & .85 & .92 & .89 & .79 & .73 & .97 & .87 & .88 & .88 & .03 \\
Ant  & .74 & .48 & .73 & .41 & .67 & .88 & .75 & .95 & .77 & .86 & .67 & .99 & .75 & .89 & .74 & .86 \\
Citeseer  & .77 & .12 & .77 & .42 & .63 & .89 & .63 & .96 & .73 & .89 & .46 & .99 & .75 & .87 & .77 & .15 \\
DBLP  & .89 & .74 & .89 & .95 & .89 & .96 & .89 & .75 & .88 & .75 & .90 & .75 & .90 & .75 & .89 & .07 \\
Facebook0  & .56 & .89 & .52 & .41 & .41 & .48 & .53 & .87 & .52 & .79 & .47 & .95 & .49 & .88 & .54 & .73 \\
Facebook686  & .53 & .98 & .51 & .38 & .52 & .53 & .48 & .85 & .47 & .87 & .50 & .96 & .53 & .88 & .54 & .91 \\
Log4J  & .70 & .70 & .70 & .39 & .66 & .85 & .64 & .97 & .67 & .97 & .65 & .99 & .69 & .96 & .70 & .60 \\
Maven  & .97 & .33 & .97 & .45 & .94 & .95 & .77 & .97 & .77 & .97 & .80 & .97 & .80 & .97 & .97 & .26 \\
Pubmed  & .83 & .36 & .77 & .42 & .67 & .95 & .71 & .97 & .75 & .90 & .58 & .98 & .75 & .87 & .82 & .59 \\
Squirrel  & .64 & .40 & .64 & .31 & .64 & .47 & .67 & .92 & .67 & .83 & .68 & .93 & .68 & .87 & .64 & .32 \\
Twitter  & .90 & .05 & .69 & .43 & .24 & .75 & .66 & .90 & .78 & .77 & .58 & .99 & .71 & .86 & .87 & .15 \\
        \hline
        \multicolumn{15}{c}{Graph Filter HK7}\\
         \hline
         ACM  & .68 & .49 & .68 & .53 & .68 & .98 & .72 & .99 & .73 & .93 & .72 & .96 & .73 & .94 & .68 & .60 \\
Amazon  & .88 & .01 & .87 & .64 & .58 & .96 & .81 & .97 & .87 & .85 & .75 & .97 & .87 & .86 & .88 & .03 \\
Ant  & .74 & .48 & .74 & .54 & .72 & .94 & .61 & .98 & .77 & .85 & .66 & .99 & .76 & .88 & .74 & .86 \\
Citeseer  & .77 & .12 & .77 & .55 & .66 & .82 & .69 & .95 & .73 & .89 & .46 & 1.00 & .75 & .86 & .77 & .15 \\
DBLP  & .89 & .74 & .89 & .85 & .89 & .93 & .90 & .69 & .85 & .69 & .90 & .69 & .90 & .69 & .89 & .07 \\
Facebook0  & .56 & .89 & .53 & .58 & .51 & .62 & .55 & .91 & .52 & .85 & .47 & .99 & .49 & .89 & .54 & .73 \\
Facebook686  & .53 & .98 & .51 & .55 & .51 & .52 & .50 & .86 & .50 & .86 & .51 & .96 & .52 & .86 & .54 & .91 \\
Log4J  & .70 & .70 & .70 & .52 & .69 & .89 & .65 & .96 & .68 & .93 & .64 & .98 & .68 & .96 & .70 & .60 \\
Maven  & .97 & .33 & .97 & .62 & .97 & .92 & .78 & .96 & .78 & .96 & .84 & .96 & .84 & .96 & .97 & .26 \\
Pubmed  & .83 & .36 & .78 & .57 & .73 & .83 & .65 & 1.00 & .75 & .92 & .58 & .99 & .74 & .85 & .82 & .59 \\
Squirrel  & .64 & .40 & .64 & .42 & .64 & .44 & .67 & .84 & .67 & .74 & .68 & .91 & .68 & .84 & .64 & .32 \\
Twitter  & .90 & .05 & .71 & .59 & .42 & .97 & .61 & .94 & .73 & .78 & .58 & .99 & .69 & .83 & .87 & .15 \\
    \end{tabular}
    \caption{Experiments for base graph filters without post-processing.}
\end{table}
\clearpage
\begin{table}[htbp]
    \footnotesize
    \setlength\tabcolsep{.5pt}
    \begin{tabular}{l | r r | r r | r r | r r | r r | r r | r r | r r}
        ~&\multicolumn{2}{c|}{\textbf{None}}&\multicolumn{2}{c|}{\textbf{Mult}}&\multicolumn{2}{c|}{\textbf{LFPRO}}&\multicolumn{2}{c|}{\textbf{FairPers}}&\multicolumn{2}{c|}{\textbf{FairPers-C}}&\multicolumn{2}{c|}{\textbf{FairEdit}}&\multicolumn{2}{c|}{\textbf{FairEdit-C}}&\multicolumn{2}{c}{\textbf{FairWalk}}\\
        ~ & AUC & pRule & AUC & pRule & AUC & pRule & AUC & pRule & AUC & pRule & AUC & pRule & AUC & pRule & AUC & pRule\\
         \hline
        \multicolumn{15}{c}{Graph Filter PPR.85S}\\
         \hline
         ACM  & .70 & .32 & .70 & .80 & .72 & .90 & .62 & .99 & .69 & .87 & .69 & .87 & .68 & .83  & .70 & .82\\
Amazon  & .79 & .01 & .76 & .82 & .57 & .93 & .71 & .96 & .82 & .82 & .67 & 1.00 & .79 & .82 & .81 & .03 \\
Ant  & .78 & .60 & .77 & .77 & .80 & .86 & .78 & .97 & .78 & .89 & .77 & .96 & .77 & .86 & .77 & .79 \\
Citeseer  & .84 & .14 & .83 & .78 & .80 & .75 & .85 & .91 & .85 & .84 & .84 & .40 & .85 & .84 & .84 & .15 \\
DBLP  & .90 & .58 & .90 & .81 & .89 & .77 & .89 & 1.00 & .76 & .79 & .88 & .77 & .88 & .76 & .92 & .09 \\
Facebook0  & .60 & .94 & .56 & .72 & .57 & .73 & .57 & .86 & .57 & .81 & .59 & .96 & .62 & .91 & .59 & .96 \\
Facebook686  & .55 & .93 & .52 & .72 & .51 & .68 & .53 & .93 & .53 & .90 & .53 & .97 & .56 & .88 & .55 & .94 \\
Log4J  & .76 & .76 & .76 & .73 & .74 & .81 & .75 & .98 & .75 & .97 & .75 & .91 & .75 & .86 & .78 & .78 \\
Maven  & 1.00 & .48 & 1.00 & .81 & 1.00 & .83 & .96 & .99 & .96 & .98 & 1.00 & .97 & 1.00 & .97 & 1.00 & .90 \\
Pubmed  & .92 & .46 & .84 & .74 & .80 & .96 & .88 & .98 & .88 & .93 & .91 & .74 & .91 & .86 & .88 & .65 \\
Squirrel  & .64 & .37 & .64 & .75 & .64 & .47 & .67 & .95 & .66 & .90 & .66 & .94 & .66 & .81 & .64 & .38 \\
Twitter  & .98 & .13 & .34 & .76 & .31 & .78 & .17 & .92 & .50 & .92 & .45 & .99 & 1.00 & .83 & .95 & .31 \\
        \hline
        \multicolumn{15}{c}{Graph Filter PPR.99S}\\
         \hline
ACM  & .68 & .37 & .68 & .91 & .73 & .90 & .56 & .99 & .61 & .87 & .68 & .89 & .67 & .81 & .68 & .86 \\
Amazon  & .79 & .01 & .77 & .96 & .64 & .83 & .61 & .96 & .64 & .84 & .65 & 1.00 & .79 & .80 & .79 & .03 \\
Ant  & .78 & .93 & .77 & .95 & .77 & .97 & .77 & .98 & .74 & .93 & .36 & .99 & .58 & .96 & .72 & .76 \\
Citeseer  & .85 & .41 & .84 & .92 & .86 & .86 & .84 & .94 & .85 & .81 & .84 & 1.00 & .85 & .81 & .85 & .40 \\
DBLP  & .89 & .54 & .89 & .79 & .88 & .74 & .84 & 1.00 & .90 & .78 & .88 & .63 & .88 & .72 & .89 & .36 \\
Facebook0  & .64 & .97 & .61 & .96 & .61 & .96 & .60 & .99 & .59 & .99 & .53 & 1.00 & .54 & .86 & .64 & .98  \\
Facebook686  & .55 & .98 & .52 & .97 & .52 & .97 & .52 & .98 & .51 & .97 & .55 & 1.00 & .48 & .98 & .51 & .96\\
Log4J  & .76 & .90 & .74 & .94 & .74 & .97 & .75 & .98 & .75 & .96 & .60 & .99 & .52 & .98 & .55 & .54 \\
Maven  & 1.00 & .73 & 1.00 & .99 & 1.00 & .99 & .97 & .98 & .96 & .93 & .98 & .97 & .99 & .94 & 1.00 & .95 \\
Pubmed  & .91 & .75 & .84 & .98 & .83 & .97 & .89 & .99 & .88 & .90 & .89 & .98 & .93 & .84 & .70 & .97 \\
Squirrel  & .64 & .47 & .64 & .89 & .64 & .64 & .62 & .93 & .62 & .77 & .62 & .95 & .62 & .78 & .64 & .49 \\
Twitter  & .98 & .60 & .43 & .97 & .40 & .97 & .65 & .99 & .96 & .83 & .55 & 1.00 & 1.00 & .80 & .73 & .15 \\
\hline
        \multicolumn{15}{c}{Graph Filter HK3S}\\
         \hline
         ACM  & .67 & .40 & .67 & .44 & .65 & .94 & .66 & .96 & .66 & .87 & .69 & .99 & .66 & .92 & .67 & .81 \\
Amazon  & .87 & .01 & .86 & .49 & .44 & .85 & .15 & .98 & .78 & .91 & .84 & 1.00 & .62 & .92 & .88 & .03 \\
Ant  & .74 & .48 & .73 & .40 & .70 & .86 & .73 & .93 & .74 & .84 & .75 & .98 & .76 & .88 & .74 & .86 \\
Citeseer  & .77 & .12 & .77 & .42 & .64 & .90 & .73 & .95 & .81 & .89 & .59 & .99 & .85 & .88 & .77 & .13 \\
DBLP  & .88 & .87 & .88 & .92 & .88 & .91 & .85 & 1.00 & .85 & .98 & .90 & .99 & .90 & .98 & .88 & .26 \\
Facebook0  & .58 & .90 & .52 & .39 & .41 & .45 & .53 & .78 & .58 & .84 & .53 & .95 & .59 & .86 & .54 & .73 \\
Facebook686  & .55 & .99 & .51 & .36 & .51 & .49 & .52 & .79 & .53 & .91 & .51 & .93 & .54 & .93 & .55 & .95  \\
Log4J  & .70 & .69 & .69 & .39 & .64 & .86 & .62 & .96 & .61 & .92 & .76 & .98 & .76 & .96 & .70 & .83 \\
Maven  & .97 & .27 & .97 & .43 & .94 & .96 & .43 & 1.00 & .43 & 1.00 & .99 & 1.00 & .99 & .99 & .97 & .26 \\
Pubmed  & .83 & .36 & .77 & .41 & .66 & .95 & .87 & .97 & .81 & .89 & .78 & .98 & .90 & .87 & .82 & .59 \\
Squirrel  & .64 & .41 & .64 & .32 & .64 & .49 & .66 & .91 & .65 & .85 & .67 & .96 & .67 & .91 & .64 & .38 \\
Twitter  & .90 & .05 & .68 & .41 & .18 & .67 & .02 & .94 & .51 & .96 & .45 & .98 & 1.00 & .88 & .88 & .10 \\
\hline
        \multicolumn{15}{c}{Graph Filter HK7S}\\
         \hline
         ACM  & .67 & .38 & .67 & .57 & .66 & .97 & .67 & .96 & .66 & .83 & .68 & .99 & .66 & .91 & .67 & .76 \\
Amazon  & .86 & .01 & .85 & .64 & .56 & .96 & .32 & .98 & .84 & .87 & .70 & 1.00 & .69 & .87 & .88 & .03 \\
Ant  & .74 & .48 & .74 & .54 & .72 & .94 & .76 & .99 & .76 & .84 & .75 & .98 & .75 & .88 & .74 & .79 \\
Citeseer  & .77 & .12 & .77 & .55 & .65 & .84 & .85 & .95 & .77 & .89 & .70 & .99 & .85 & .87 & .77 & .13 \\
DBLP  & .88 & .77 & .88 & .91 & .88 & .82 & .85 & .98 & .85 & .92 & .90 & .97 & .90 & .96 & .88 & .32 \\
Facebook0  & .59 & .91 & .54 & .55 & .52 & .57 & .57 & .86 & .60 & .87 & .51 & 1.00 & .61 & .85 & .54 & .73 \\
Facebook686  & .55 & .99 & .52 & .53 & .50 & .50 & .55 & .95 & .56 & .86 & .53 & .94 & .55 & .96 & .55 & .95 \\
Log4j  & .70 & .69 & .69 & .53 & .67 & .90 & .72 & .93 & .70 & .86 & .75 & .98 & .76 & .95 & .70 & .86 \\
Maven  & .97 & .25 & .97 & .60 & .97 & .89 & .70 & .99 & .70 & .99 & .99 & .98 & .99 & .98 & .97 & .26 \\
Pubmed  & .83 & .36 & .78 & .56 & .67 & .96 & .86 & 1.00 & .88 & .87 & .79 & .98 & .89 & .86 & .82 & .59 \\
Squirel  & .64 & .42 & .64 & .43 & .64 & .45 & .67 & .92 & .65 & .79 & .66 & .96 & .67 & .88 & .64 & .39 \\
Twitter  & .90 & .05 & .71 & .56 & .35 & .86 & .41 & .96 & .87 & .83 & .47 & .98 & 1.00 & .86 & .88 & .10 \\
    \end{tabular}
    \caption{Experiments for base graph filters with sweep ratio post-processing.}
\end{table}

\clearpage

\section{Ablation Study}\label{ablation study}
In this appendix we investigate the effect of adding granular retention of original posteriors in Equation~\ref{FairEdit}. To this end, we perform additional experiments on the FairEdit and FairEdit-C approaches, in which we fix the parameter value $a_0=0$ to remove explicit retention of original priors---and hence original posteriors. In Table~\ref{tab:ablation} we present the outcome of the two new variations under the names FairEdit0 and FairEdit0-C respectively. We remind that FairPers0 variants differ from FairEdit in that they use error-based instead of difference-based skewing of posteriors and KL-divergence from original posteriors as a training objective instead of mean absolute error. And FairPers variants further differ from FairPers0 in that they explicitly introduce a degree of freedom towards maintain original posteriors.
\par
Table~\ref{tab:comparison ablation} summarizes a multiway comparison between prior editing approaches. Overall, all three types of prior editing yield similar average AUC and pRule values and Nemenyi ranks when used to satisfy the 80\% pRule constraint. On the other hand FairEdit0 manages to improve AUC by only a small margin ($\sim$1\%) compared to FairPers when maximizing fairness. This indicates that concerns over posterior outliers and the interpretability of stochastic surrogate models are valid yet not too prevalent in the graphs we experiment on. At the same time, small improvements do not retract from the added value of accounting for these phenomena, as posterior outliers could arise in different graphs and the more intuitive interpretation of prior editing parameters can prove useful in terms of explainability. Finally, FairEdit's 3-4\% AUC improvement compared to FairPers can be in large part be attributed to partially retaining priors, as recommended by Theorem~\ref{anyedit}.

\begin{table}[b!]
\footnotesize
\centering
    \begin{tabular}{l | r r r | r r r}
        & \multicolumn{3}{c|}{\textbf{No Post-processing}} & \multicolumn{3}{c}{\textbf{Sweep Ratio}}\\
          & \textbf{AUC} & \textbf{pRule} & \textbf{pRule $\geq$ 80\%} & \textbf{AUC} & \textbf{pRule} & \textbf{pRule $\geq$ 80\%} \\
          \hline
         FairPers & .66 (4.5) & .93 (2.5) & .92 & .66 (5.3) & .95 (2.7) & .96\\
         FairPers-C & .72 (3.2) & .86 (4.6) & .79 & .72 (4.7) & .88 (4.4) & .92\\
         FairEdit & .69 (4.0) & .92 (2.4) & .88 & .70 (4.6) & .94 (2.4) & .92\\
         FairEdit-C & .72 (2.7) & .86 (4.4) & .90 & .77 (3.3) & .88 (4.5) & .94\\
         FairEdit0 & .67 (4.1) & .93 (2.3) & .88 & .67 (4.6) & .95 (2.1) & .92\\
         FairEdit0-C & .72 (2.5) & .86 (4.6) & .90 & .77 (3.3) & .87 (4.9) & .94\\
    \end{tabular} 
    \label{tab:comparison ablation}
    \caption{Average AUC, pRule and fraction of experiments achieving $80\%$ pRule for prior editing approaches (higher are better). Average Nemenyi ranks for measures in parenthesis (smaller are better), where rank differences greater than 1.3 are statistically significant. Different results are presented depending on the post-processeing of the base graph filter.}
\end{table}

\begin{table}[tbp]
    \centering
    \footnotesize
    \setlength\tabcolsep{.5pt}
    \begin{tabular}{l | r r | r r | r r | r r}
    & \multicolumn{2}{c|}{FairEdit0} & \multicolumn{2}{c|}{FairEdit0-C} & \multicolumn{2}{c|}{FairEdit0} & \multicolumn{2}{c}{FairEdit0-C}\\
    & AUC & pRule & AUC & pRule & AUC & pRule & AUC & pRule \\
    \hline
    & \multicolumn{4}{c|}{PPR.85} & \multicolumn{4}{c}{PPR.85S}\\ 
    \hline
Acm  & .75 & .81 & .73 & .94 & .69 & .82 & .68 & .83 \\
Amazon  & .69 & .98 & .73 & .83 & .64 & 1.00 & .80 & .82 \\
Ant  & .75 & .96 & .78 & .88 & .77 & .97 & .77 & .86 \\
Citeseer  & .87 & .42 & .82 & .82 & .84 & .40 & .85 & .83 \\
Dblp  & .89 & .66 & .89 & .66 & .88 & .77 & .88 & .76 \\
Facebook0  & .54 & .98 & .58 & .88 & .54 & .99 & .62 & .91 \\
Facebook686  & .50 & .99 & .52 & .91 & .50 & 1.00 & .56 & .89 \\
Log4j  & .74 & .93 & .76 & .92 & .75 & .91 & .75 & .86 \\
Maven  & .85 & .91 & .85 & .91 & 1.00 & .97 & 1.00 & .97 \\
Pubmed  & .86 & .69 & .79 & .86 & .91 & .74 & .91 & .85 \\
Squirel  & .75 & .87 & .75 & .85 & .66 & .94 & .66 & .81 \\
Twitter  & .58 & .99 & .68 & .82 & .58 & 1.00 & 1.00 & .80 \\
    \hline
    & \multicolumn{4}{c|}{PPR.99} & \multicolumn{4}{c}{PPR.99S}\\ 
    \hline
Acm  & .73 & .97 & .73 & .90 & .67 & .87 & .68 & .81 \\
Amazon  & .69 & 1.00 & .72 & .81 & .64 & 1.00 & .79 & .81 \\
Ant  & .53 & .99 & .54 & .97 & .31 & .99 & .39 & .97 \\
Citeseer  & .48 & .99 & .82 & .81 & .47 & 1.00 & .85 & .81 \\
Dblp  & .92 & .68 & .91 & .69 & .88 & .63 & .88 & .72 \\
Facebook0  & .52 & .99 & .47 & .88 & .53 & 1.00 & .51 & .85 \\
Facebook686  & .51 & .99 & .51 & .98 & .51 & 1.00 & .48 & .98 \\
Log4j  & .66 & .97 & .56 & .94 & .61 & .99 & .50 & .98 \\
Maven  & .87 & .83 & .89 & .80 & .98 & .97 & .99 & .94 \\
Pubmed  & .54 & 1.00 & .75 & .85 & .89 & .98 & .93 & .84 \\
Squirel  & .71 & .90 & .70 & .73 & .62 & .95 & .62 & .78 \\
Twitter  & .59 & .99 & .69 & .81 & .55 & 1.00 & 1.00 & .80 \\
    \hline
    & \multicolumn{4}{c|}{HK3} & \multicolumn{4}{c}{HK3S}\\ 
    \hline
Acm  & .72 & .95 & .73 & .93 & .62 & 1.00 & .66 & .91 \\
Amazon  & .76 & .99 & .88 & .87 & .67 & 1.00 & .88 & .84 \\
Ant  & .70 & 1.00 & .76 & .88 & .67 & .99 & .76 & .88 \\
Citeseer  & .45 & 1.00 & .76 & .86 & .46 & 1.00 & .85 & .87 \\
Dblp  & .90 & .75 & .90 & .75 & .90 & .99 & .90 & .98 \\
Facebook0  & .47 & 1.00 & .52 & .81 & .47 & 1.00 & .59 & .82 \\
Facebook686  & .50 & .98 & .53 & .89 & .47 & .97 & .55 & .93 \\
Log4j  & .65 & .99 & .69 & .96 & .76 & .98 & .76 & .96 \\
Maven  & .79 & .97 & .79 & .97 & .99 & 1.00 & .99 & .99 \\
Pubmed  & .53 & 1.00 & .76 & .86 & .53 & 1.00 & .91 & .86 \\
Squirel  & .68 & .93 & .68 & .87 & .67 & .96 & .67 & .91 \\
Twitter  & .59 & .99 & .74 & .83 & .64 & 1.00 & 1.00 & .84 \\
    \hline
    & \multicolumn{4}{c|}{HK7} & \multicolumn{4}{c}{HK7S}\\ 
    \hline
Acm  & .72 & .96 & .73 & .93 & .62 & 1.00 & .66 & .91 \\
Amazon  & .77 & .99 & .88 & .85 & .67 & 1.00 & .87 & .83 \\
Ant  & .65 & .98 & .76 & .88 & .68 & 1.00 & .75 & .88 \\
Citeseer  & .44 & 1.00 & .76 & .85 & .46 & 1.00 & .85 & .86 \\
Dblp  & .90 & .69 & .90 & .69 & .90 & .97 & .90 & .96 \\
Facebook0  & .47 & .99 & .54 & .83 & .47 & 1.00 & .63 & .84 \\
Facebook686  & .50 & .99 & .52 & .86 & .47 & .98 & .55 & .96 \\
Log4j  & .64 & .98 & .68 & .96 & .75 & .97 & .76 & .94 \\
Maven  & .84 & .96 & .84 & .96 & .99 & .98 & .99 & .98 \\
Pubmed  & .54 & 1.00 & .75 & .85 & .53 & 1.00 & .90 & .85 \\
Squirel  & .68 & .92 & .68 & .84 & .66 & .96 & .67 & .88 \\
Twitter  & .58 & .99 & .71 & .81 & .64 & 1.00 & 1.00 & .83 \\
    \end{tabular}
    \label{tab:ablation}
    \caption{Experiments of FairEdit with $a_0=0$ for both types of post-processing }
\end{table}

\clearpage

\bibliography{bibliography}{}

\end{document}